\pdfoutput=1
\documentclass[onecolumn,a4paper,nopdfoutputerror]{quantumarticle}

\usepackage[utf8]{inputenc}
\usepackage{amsmath,amssymb,mathtools}
\usepackage{array}
\usepackage{amsthm}
\usepackage{graphicx}
\usepackage{xcolor}
\usepackage{ragged2e}
\usepackage{tikz}
\usetikzlibrary{positioning,arrows.meta,calc}
\usepackage[numbers,sort&compress]{natbib}
\usepackage{hyperref}

\theoremstyle{plain}
\makeatletter
\@ifundefined{theorem}{\newtheorem{theorem}{Theorem}}{}
\@ifundefined{lemma}{\newtheorem{lemma}{Lemma}}{}
\@ifundefined{corollary}{\newtheorem{corollary}{Corollary}}{}
\@ifundefined{proposition}{\newtheorem{proposition}{Proposition}}{}
\makeatother

\newcommand{\Tr}{\operatorname{Tr}}
\newcommand{\Swap}{S}
\newcommand{\EAC}{\mathcal C_{\mathrm{abs}}}
\newcommand{\EACF}{\mathcal C_F}
\newcommand{\PhiP}{\Phi^{+}}
\newcommand{\PhiM}{\Phi^{-}}
\newcommand{\PsiP}{\Psi^{+}}
\newcommand{\PsiM}{\Psi^{-}}
\newcommand{\Sep}{\mathrm{Sep}}
\newcolumntype{C}[1]{>{\centering\arraybackslash}p{#1}}

\begin{document}

\title{Absorption capacity of separable noise: Bell-mixing thresholds on
entanglement and teleportation}

\author{Xuan Du Trinh}
\affiliation{Stony Brook University, Stony Brook, New York 11794, USA}
\email{xtrinh@cs.stonybrook.edu}
\orcid{0009-0009-5610-462X}

\begin{abstract}
We study Bell-mixing lines
$\rho_\lambda=\lambda\PhiP+(1-\lambda)\sigma$, where $\PhiP$ is a fixed Bell
reference and $\sigma$ is a separable two-qubit noise state. Along this line
there are two operational crossings: the state becomes entangled, and it
reaches quantum teleportation advantage over classical strategies. We
package these crossings as capacities of the noise state. The entanglement
absorption capacity $\EAC(\sigma)$ is the
largest amount of Bell reference that $\sigma$ can absorb while the partial
transpose remains positive. The fidelity absorption capacity $\EACF(\sigma)$
is the largest amount of Bell reference that $\sigma$ can absorb while
keeping the maximal teleportation fidelity at or below the classical bound
$2/3$. The thresholds corresponding to the two crossing points are obtained
from the same M\"obius map, $\lambda_*=\EAC/(1+\EAC)$ and
$\lambda_F=\EACF/(1+\EACF)$. We derive closed-form capacities and
thresholds for product noise states and separable complex $X$ noise states.
For product noise, $\EAC$ depends only on local
marginal purities, while $\EACF$ also depends on orientation relative to the
maximally entangled reference. For $X$ noise states, both capacities are
explicit in all four Bell frames. We also study three extensions:
arbitrary pure-state references, the evolution of $X$ noise states and their
capacities under local amplitude-damping and dephasing channels, and
decomposition certificates that give lower bounds on the capacities, hence
on the thresholds, for general separable noise.
\end{abstract}

\maketitle

\section{Introduction}
\label{sec:intro}
Entanglement is fragile. It is the resource behind protocols from
teleportation~\cite{BennettTeleportation1993} to entanglement-based key
distribution~\cite{Ekert1991}, but in any real device the entangled state
must first survive preparation, storage, and transmission through a noisy
environment. How that entanglement changes, and when it is finally gone,
is a basic question of open-system dynamics. The loss need not be
asymptotic: under common noise processes a two-qubit state can become
separable at a finite time, the phenomenon of
entanglement sudden death (ESD)~\cite{YuEberly2004,AlmeidaEtAl2007Science,
YuEberly2009Science,TerraCunha2007,AolitaMeloDavidovich2015}. A second
transition matters for applications: a state may remain entangled after it
has already stopped beating the classical teleportation bound. This paper
gives a framework for locating both transitions along a Bell-mixing line
and following them as the noise state evolves.

We write the noisy state as a mixture of the intended Bell component and a
separable background,
\begin{equation}
  \rho_\lambda \,=\, \lambda\,|\PhiP\rangle\!\langle\PhiP|
  + (1-\lambda)\,\sigma,\qquad \lambda\in[0,1],
\label{eq:bell-mixing-line}
\end{equation}
where $|\PhiP\rangle=(|00\rangle+|11\rangle)/\sqrt2$, $\sigma$ is
separable, and $\lambda$ is the surviving Bell weight. This line is not
only a parametrization. It is the output of a replacement process that
transmits the input pair with probability $\lambda$ and otherwise prepares
$\sigma$. It also describes depolarization toward a fixed separable
endpoint. In this parametrization, $\lambda$ has the meaning of a survival probability, and the
line runs from the noise state $\rho_0=\sigma$ to the Bell reference
$\rho_1=|\PhiP\rangle\!\langle\PhiP|$.

The first transition is the entanglement threshold
\begin{equation}
\lambda_*(\sigma)\,:=\,\inf\{\lambda\in[0,1]:\rho_\lambda\text{ is entangled}\},
\label{eq:lambda-star-def}
\end{equation}
the smallest Bell weight above which the mixture is entangled. On two
qubits, separability is equivalent to positive partial transpose (PPT) by
the Peres--Horodecki criterion~\cite{Peres1996,Horodecki1996}. The second
transition is the teleportation threshold
\begin{equation}
\lambda_F(\sigma)\,:=\,\inf\{\lambda\in[0,1]:F(\rho_\lambda)>2/3\},
\label{eq:lambda-F-def}
\end{equation}
where $F(\rho_\lambda)$ is the maximal teleportation fidelity that
$\rho_\lambda$ can provide, and $2/3$ is the classical
bound~\cite{Popescu1994,MassarPopescu1995,HHH96tele}. At the
Werner point $\sigma=I_4/4$ the two thresholds coincide,
$\lambda_*=\lambda_F=1/3$~\cite{Werner1989,Peres1996,Horodecki1996}. Away
from that point, the useful object is the map from the noise state to
the two transition weights, $\sigma\mapsto\lambda_*(\sigma)$ and
$\sigma\mapsto\lambda_F(\sigma)$. For calibration, or for a
time-dependent channel, this map turns a measured or evolving noise state
$\sigma(t)$ directly into the Bell weight at which each qualitative change
occurs.

Several physically motivated noise families make this map concrete.
Bell-diagonal noise is the standard reference case, with separability fixed by
the largest Bell
weight~\cite{BennettDiVincenzoSmolinWootters1996,HorodeckiRMP}. Product
noise $\sigma=A\otimes B$ describes independent local decoherence of the
two qubits, although the dependence of the threshold on the two local
marginals is not obvious from the pointwise criteria. We also consider
$X$-state noise, where the separable endpoint has the form
$$
\sigma_X=\begin{pmatrix}
a&0&0&u\\ 0&b&v&0\\ 0&v^{*}&c&0\\ u^{*}&0&0&d
\end{pmatrix}
\qquad(a,b,c,d\ge0,\ \ u,v\in\mathbb{C}),
$$
with nonzero entries only on the diagonal and antidiagonal in the
computational basis. This family is closed under partial
transposition~\cite{YuEberly2007,AliRauAlber2010,QuesadaAlQasimiJames2012}
and carries explicit entanglement and teleportation
criteria~\cite{Hu2013XStates}. It is also preserved by the local
amplitude-damping and dephasing channels used in standard relaxation
models, and includes the setting where ESD was first
identified~\cite{YuEberly2004,TerraCunha2007}.

Our approach works from the separable side. Instead of asking how much
noise destroys a fixed entangled state, we ask how much of the Bell
reference a fixed separable background can absorb before a transition is
reached. The two resulting scalars are the \emph{entanglement absorption
capacity} $\EAC(\sigma)$ and the \emph{fidelity absorption capacity}
$\EACF(\sigma)$. Writing $\eta=\lambda/(1-\lambda)$, the PPT condition for
\eqref{eq:bell-mixing-line} becomes positivity of
$\sigma^{T_B}+\eta Q$, where
$Q=(|\PhiP\rangle\!\langle\PhiP|)^{T_B}$. We define $\EAC(\sigma)$ as the
largest feasible $\eta$ for this positivity problem. We define
$\EACF(\sigma)$ with the same mixing parameter, replacing PPT positivity by
the condition that the fully entangled fraction is still at most $1/2$,
equivalently that the state does not beat the classical teleportation
fidelity $2/3$ for two qubits~\cite{HHH96tele}. After these capacities are
computed in the $\eta$ parameter, the Bell-weight thresholds are obtained by
converting back to $\lambda$:
$$
\lambda_*(\sigma)=\frac{\EAC(\sigma)}{1+\EAC(\sigma)},
\qquad
\lambda_F(\sigma)=\frac{\EACF(\sigma)}{1+\EACF(\sigma)}.
$$

This construction is directional. The noise state $\sigma$ is the input to
the capacities, and for each input we test a single mixing line from
$\sigma$ to the chosen reference. In \eqref{eq:bell-mixing-line} the
reference is $|\PhiP\rangle$. Each capacity is therefore a threshold value
along a specified line, not a direction-independent measure of $\sigma$.
What is new is the closed-form treatment of both threshold crossings in one joint
framework, not either quantity alone. Both capacities are
concave under convex decomposition of the separable noise
as proved below, so decompositions into solved components give certified
lower bounds.
Table~\ref{tab:prior-positioning} places this viewpoint against the
case-by-case treatments.

For a general separable endpoint, $\EAC$ is still obtained from a small
SDP. The structured families below are useful because they turn the exact
threshold, or a certified lower bound from a convex decomposition into
solved pieces, into a formula. This is often enough to identify ESD or the
loss of teleportation advantage. The point is not computational speed,
since the SDP is small, but analytic control: a closed threshold can be
bounded, compared, differentiated, and followed along a channel.

\paragraph*{Summary of results.}
For product noise, $\EAC$ depends on the two
marginals only through their
purities,
\begin{equation}
\EAC(A\otimes B)=\sqrt{(1-\Tr A^2)(1-\Tr B^2)},
\label{eq:intro-product}
\end{equation}
and this value is the same for every maximally entangled reference. By
contrast, $\EACF$ is not fixed by these purities alone: it also depends on
how the product state is oriented relative to the chosen reference. The
separable $X$ family gives closed forms for both capacities
in all four Bell frames, and therefore for $\lambda_*$, $\lambda_F$, and
the gap $\lambda_F-\lambda_*$. For a pure reference with concurrence
$C_\psi\in(0,1]$~\cite{Wootters1998}, $\EAC$ rescales by $1/C_\psi$ for
every product noise. On $X$ noise, the pure-reference rescaling used here
applies when the noise state is $X$-shaped in the Schmidt basis of the
reference. Because amplitude damping and dephasing
preserve $X$ form, substituting the time-dependent $X$ entries gives
explicit threshold curves $(\lambda_*(t),\lambda_F(t))$ along a
decoherence trajectory. Section~\ref{sec:summary} states the laws as
theorems, and Section~\ref{subsec:using-framework} gives the recipe to
apply them.

\section{The framework}
\label{sec:eac}
This section defines the framework's two absorption capacities, the
entanglement absorption capacity $\EAC(\sigma)$ and the fidelity
absorption capacity $\EACF(\sigma)$, from the Peres--Horodecki and
teleportation criteria. It proves the M\"obius identity that turns $\EAC$
into the entanglement threshold and $\EACF$ into the teleportation
threshold, and states the closed-form theorems that follow on structured
noise families. The
remaining sections derive the product-noise, $X$-state, and channel
trajectory formulas.

\subsection{Definition and operational meaning}
We use the standard Bell basis
\begin{equation}
|\Phi^\pm\rangle=\frac{|00\rangle\pm|11\rangle}{\sqrt2},
\qquad
|\Psi^\pm\rangle=\frac{|01\rangle\pm|10\rangle}{\sqrt2}.
\label{eq:Bell-basis}
\end{equation}
The same symbols without kets denote the corresponding rank-one
projectors, e.g.\ $\PhiP=|\PhiP\rangle\!\langle\PhiP|$, and a reference
drawn from $\Phi^\pm,\Psi^\pm$ is called a \emph{Bell frame}. The paper
works throughout with two qubits and takes $|\PhiP\rangle$ as the default
reference, except where a different Bell frame or a pure-state reference is
stated explicitly. The noise state is assumed separable whenever a
Bell-mixing entanglement threshold is quoted. We define
\begin{equation}
Q:=(|\PhiP\rangle\!\langle\PhiP|)^{T_B},
\qquad
\Swap:=\sum_{i,j=0}^{1}|ij\rangle\!\langle ji|.
\label{eq:Q-swap-identity}
\end{equation}
Here $Q$ denotes the partial transpose of the projector onto the Bell
reference state, while
$\Swap$ denotes the standard swap (or flip) operator on
$\mathbb{C}^2\otimes\mathbb{C}^2$. We refer to $Q$ as the
\emph{Bell-reference partial-transpose operator}. Throughout, it serves as the
fixed direction in which the noise state absorbs the Bell component as
the entanglement absorption capacity is computed. The swap identity $Q=\Swap/2$ is the
key step in the $U\otimes U$ symmetry argument in the product-state law of
Section~\ref{sec:eval1}.

\paragraph*{Entanglement absorption.} Partial transposition turns the mixture
$\rho_\lambda=\lambda|\PhiP\rangle\!\langle\PhiP|+(1-\lambda)\sigma$
into $\rho_\lambda^{T_B}=\lambda Q+(1-\lambda)\sigma^{T_B}$, so the
separability question (when does $\rho_\lambda^{T_B}$ stay positive?)
is a one-parameter positivity problem along the fixed direction $Q$ starting
from $\sigma^{T_B}$. We summarize that positivity problem by a single
number:
\begin{equation}
\EAC(\sigma):=\sup\{\eta\ge 0:\ \sigma^{T_B}+\eta Q\succeq 0\},\qquad \sigma\in\Sep,
\label{eq:EAC-def}
\end{equation}
i.e.\ the largest coefficient $\eta$ of $Q$ that $\sigma^{T_B}$ can
absorb while remaining positive semidefinite.
Restricting the parameter to $\eta\ge0$ makes $\EAC$ a \emph{non-negative}
capacity on its natural domain, the separable set $\Sep$: a separable
$\sigma$ has $\sigma^{T_B}\succeq0$, so $\eta=0$ is feasible and
$\EAC(\sigma)\ge0$. The feasible set in \eqref{eq:EAC-def} is a closed
bounded interval $[0,\EAC(\sigma)]$, so the supremum is attained and finite
(proof of Theorem~\ref{thm:mobius}). The value $\EAC(\sigma)=0$ means that
$\sigma$ cannot absorb any positive amount of the Bell component $Q$.

\paragraph*{Fidelity absorption.} The framework's second
question, teleportation usefulness, is answered by the same one-parameter
construction with a different stopping condition. At parameter $\eta$,
the Bell-mixing point is
$$
\rho_\eta=\frac{\sigma+\eta\,\PhiP}{1+\eta},
\qquad
\lambda=\frac{\eta}{1+\eta}.
$$
For a two-qubit state $\rho$, let
$$
f(\rho)=\max_{|\beta\rangle\in\mathcal M}
\langle\beta|\rho|\beta\rangle
$$
be the fully entangled fraction, where $\mathcal M$ is the set of maximally
entangled states. Since the maximal teleportation fidelity is
$F(\rho)=(2f(\rho)+1)/3$~\cite{HHH96tele}, the classical bound
$F\le2/3$ is equivalent to $f\le1/2$. We define the
\emph{fidelity absorption capacity}
\begin{equation}
\EACF(\sigma):=\sup\bigl\{\eta\ge0:\ f(\rho_\eta)\le\tfrac12\bigr\},
\label{eq:EACF-def}
\end{equation}
the largest amount of the Bell reference that $\sigma$ can absorb before
the mixture beats the classical teleportation bound. As with $\EAC$, this
supremum is well posed: the feasible set is a closed interval
$[0,\EACF(\sigma)]$ containing $\eta=0$, so $\EACF(\sigma)$ is attained and
finite, as shown in the proof of Theorem~\ref{thm:mobius}. The two capacities stand on equal
footing: $\EAC$ decides whether the mixture is entangled and $\EACF$
whether it is useful for teleportation, and the next subsection turns each
into its Bell-weight threshold by the same M\"obius map. The closed form
of $\EACF$, like that of $\EAC$, is stated with the other results in
Section~\ref{sec:summary}. Its derivation, and the precise sense in which
\eqref{eq:EACF-def} is the first teleportation crossing, are given in
Sections~\ref{sec:eval1} and~\ref{sec:eval2}.

\paragraph*{Reference convention.} We keep the reference
implicit in the notation: $\EAC(\sigma)\equiv\EAC(\sigma;\PhiP)$, and we
show it explicitly, as $\EAC(\sigma;\beta)$ for another Bell frame
$\beta\in\{\Phi^\pm,\Psi^\pm\}$ or $\EAC(\sigma;\psi)$ for a pure-state
reference, only when it is not $|\PhiP\rangle$. The same applies to
$\EACF$.

\subsection{The M\"obius threshold map}

\begin{theorem}[M\"obius threshold identities]
\label{thm:mobius}
For every separable two-qubit state $\sigma$, the Bell-mixing
entanglement threshold is the M\"obius transform
$\eta\mapsto \eta/(1+\eta)$ of the entanglement absorption capacity, and
the teleportation threshold is the same transform of the fidelity
absorption capacity,
\begin{equation}
\lambda_*(\sigma)=\frac{\EAC(\sigma)}{1+\EAC(\sigma)},\qquad
\lambda_F(\sigma)=\frac{\EACF(\sigma)}{1+\EACF(\sigma)}.
\label{eq:mobius}
\end{equation}
\end{theorem}
\begin{proof}[Proof]
On two qubits, the Peres--Horodecki criterion identifies separability
with positive partial transpose~\cite{Peres1996,Horodecki1996}. Hence
the threshold in \eqref{eq:lambda-star-def} is the largest
$\lambda$ for which $\rho_\lambda^{T_B}$ remains positive semidefinite. For
$\sigma\in\Sep$, $\sigma^{T_B}\succeq 0$, so the feasible set in
\eqref{eq:EAC-def} contains $\eta=0$, and since $Q=\Swap/2$ has its only
negative eigenvalue on $|\PsiM\rangle$, the constraint
$\langle\PsiM|(\sigma^{T_B}+\eta Q)|\PsiM\rangle\ge0$ bounds $\eta$ above,
so this feasible set is a closed bounded interval $[0,\EAC(\sigma)]$ with
the supremum attained. The partial transpose of the
mixing line is
$\rho_\lambda^{T_B}=\lambda Q+(1-\lambda)\sigma^{T_B}$, and (for
$\lambda<1$) the condition $\rho_\lambda^{T_B}\succeq 0$ is
equivalent (after dividing by $1-\lambda$) to
$\sigma^{T_B}+\tfrac{\lambda}{1-\lambda}Q\succeq 0$, i.e., to
$\lambda/(1-\lambda)\le\EAC(\sigma)$. The change of variable
$\eta=\lambda/(1-\lambda)$ is monotone increasing from $[0,1)$ onto
$[0,\infty)$ with inverse $\lambda=\eta/(1+\eta)$. Solving for
$\lambda$ at the crossing gives the first identity in \eqref{eq:mobius}.
The second follows by the same change of variables with the stopping
condition $f(\rho_\lambda)\le\tfrac12$ in place of PPT positivity. Since
$\rho_\lambda$ is affine in $\lambda$, each overlap
$\lambda\mapsto\langle\beta|\rho_\lambda|\beta\rangle$ is affine, so the
fully entangled fraction
$f(\rho_\lambda)=\max_{|\beta\rangle\in\mathcal M}\langle\beta|\rho_\lambda|\beta\rangle$
is convex in $\lambda$ on $[0,1]$, as a pointwise maximum of affine
functions of $\lambda$. By convexity, if
$f(\rho_{\lambda_1}),f(\rho_{\lambda_2})\le\tfrac12$ then
$f(\rho_{t\lambda_1+(1-t)\lambda_2})\le
t\,f(\rho_{\lambda_1})+(1-t)\,f(\rho_{\lambda_2})\le\tfrac12$ for all
$t\in[0,1]$, so $\{\lambda\in[0,1]:f(\rho_\lambda)\le\tfrac12\}$, as a
sublevel set of the convex function $f(\rho_\lambda)$, is convex, hence an
interval. It contains $\lambda=0$, where
$f(\sigma)\le\tfrac12$ because $\sigma$ is separable, and excludes
$\lambda=1$, where $f(\PhiP)=1$, so it equals $[0,\lambda_F]$ with right
endpoint $\lambda_F=\inf\{\lambda:f(\rho_\lambda)>\tfrac12\}$. In the
variable $\eta=\lambda/(1-\lambda)$ the same endpoint is
$\sup\{\eta\ge0:f(\rho_\eta)\le\tfrac12\}=\EACF(\sigma)$, so the monotone
inverse $\lambda=\eta/(1+\eta)$ gives $\lambda_F=\EACF/(1+\EACF)$.
\end{proof}
Theorem~\ref{thm:mobius} is the structural identity that the rest of the
paper specializes: the two operational thresholds are one M\"obius map
applied to the two capacities.
\subsection{Main results on closed forms}
\label{sec:summary}
Having defined the two capacities and proved the M\"obius identity
\eqref{eq:mobius}, the closed-form results for the two
capacities are now stated for two structured families, product states and
$X$ states. Table~\ref{tab:prior-positioning} compares this viewpoint
against prior uses of the same underlying criteria.

\emph{Product noise} $\sigma=A\otimes B$ models independent local
decoherence of the two qubits, with single-qubit marginals
$A=(I+\vec a\!\cdot\!\boldsymbol\sigma)/2$ and
$B=(I+\vec b\!\cdot\!\boldsymbol\sigma)/2$ of Bloch vectors
$\vec a,\vec b$. The $X$-state notation used below is
\begin{equation}
\sigma_X=\begin{pmatrix}
a&0&0&u\\ 0&b&v&0\\ 0&v^{*}&c&0\\ u^{*}&0&0&d
\end{pmatrix},\qquad
\begin{gathered}
a,b,c,d\ge 0,\\
a+b+c+d=1,\\
ad\ge |u|^2,\\
bc\ge |v|^2 .
\end{gathered}
\label{eq:summary-Xstate}
\end{equation}
Here $a,b,c,d$ are the diagonal populations and $u,v$ are the
antidiagonal coherences. For this $X$ form, the conditions in
\eqref{eq:summary-Xstate} are exactly $\Tr\sigma_X=1$ and
$\sigma_X\succeq0$. Separability adds the PPT constraints
$ad\ge|v|^2$ and $bc\ge|u|^2$. In these two additional constraints, the
interchange of $u$ and $v$ compared with the positivity conditions in
\eqref{eq:summary-Xstate} is intentional: partial transposition swaps the two
antidiagonal coherences.
We call the
$\{|00\rangle,|11\rangle\}$ subspace the $\Phi$ block and the
$\{|01\rangle,|10\rangle\}$ subspace the $\Psi$ block, since they support
the Bell states $\Phi^\pm$ and $\Psi^\pm$, respectively. The $\Phi$ block
carries $u$ and the $\Psi$ block carries $v$.

\begin{theorem}[Product law]
\label{thm:product}
For product noise the entanglement absorption capacity is the geometric
mean of the two single-qubit impurities,
\begin{equation*}
\EAC(A\otimes B)=\sqrt{(1-\Tr A^2)(1-\Tr B^2)},
\end{equation*}
the same for every maximally entangled reference. Its threshold obeys
$\lambda_*(A\otimes B)\le\tfrac13$, with equality if and only if
$A=B=I/2$. The fidelity absorption capacity also has a closed form, but it
additionally records the product state's orientation relative to the reference: in the
$|\PhiP\rangle$ frame, with $h=|\vec a|\,|\vec b|$ and
$m=\vec a^{T}\!\operatorname{diag}(1,-1,1)\,\vec b$,
\begin{equation*}
\EACF(A\otimes B)=\frac{\sqrt{1-h^2+m^2}-m}{2}.
\end{equation*}
The corresponding formulas for the other Bell frames are given in
Section~\ref{sec:eval1}. A pure-state reference
of concurrence $C_\psi\in(0,1]$ rescales the entanglement absorption capacity
universally, $\EAC(A\otimes B;\psi)=\EAC(A\otimes B)/C_\psi$.
\textnormal{(Derived in Section~\ref{sec:eval1}.)}
\end{theorem}

\begin{theorem}[$X$-state absorption capacity]
\label{thm:xstate}
For a separable $X$ state $\sigma_X$, the entanglement absorption capacity in the
$|\PhiP\rangle$ frame is
\begin{equation*}
\EAC(\sigma_X;\Phi^+)=-2\operatorname{Re}u+2\sqrt{bc-(\operatorname{Im}u)^2},
\end{equation*}
The other
three Bell frames have the parallel closed forms
\eqref{eq:EAC-X-PhiM}--\eqref{eq:EAC-X-PsiM}, and the Bell-diagonal and
product laws with diagonal marginals follow as special cases. If
$|\psi\rangle=\alpha|00\rangle+\beta|11\rangle$ and the noise state is
$X$-shaped in this Schmidt basis, then
$\EAC(\sigma_X;\psi)=\EAC(\sigma_X;\Phi^+)/C_\psi$.
\textnormal{(Derived in Section~\ref{sec:eval2}.)}
\end{theorem}

\begin{theorem}[Teleportation threshold and gap]
\label{thm:xstate-tele}
On the separable $X$ sector the fidelity absorption capacity in the
$|\PhiP\rangle$ frame is
\begin{equation*}
\EACF(\sigma_X;\Phi^+)=-2\operatorname{Re}u+\sqrt{(b+c)^2-4(\operatorname{Im}u)^2},
\end{equation*}
The other three Bell frames have the parallel closed forms
\eqref{eq:EACF-X-PhiM}--\eqref{eq:EACF-X-PsiM}. If
$|\psi\rangle=\alpha|00\rangle+\beta|11\rangle$ and the noise state is
$X$-shaped in this Schmidt basis, then
$\EACF(\sigma_X;\psi)=\EACF(\sigma_X;\Phi^+)/C_\psi$. The
inequality $\lambda_*\le\lambda_F$ holds for every separable two-qubit
noise state. On the $X$ sector, the equality $\lambda_*=\lambda_F$ holds
if and only if the two middle populations coincide, $b=c$, and the gap is
given explicitly in Corollary~\ref{cor:sandwich-gap}.
\textnormal{(Derived in Section~\ref{sec:eval2}.)}
\end{theorem}

Both the product and $X$ families are preserved by the local relaxation
and dephasing channels considered below, so their threshold trajectories
are obtained by substituting the channel-evolved entries into the
closed-form formulas for $\EAC$ and $\EACF$. Sections~\ref{sec:eval1}
and~\ref{sec:eval2} derive the product and $X$-state laws,
Section~\ref{sec:trajectories} applies them to decoherence trajectories,
and Section~\ref{sec:outlook} collects the open directions. The
product-family $\EAC$ trajectories are derived in
Appendix~\ref{sec:supp-channel-calculus}.

\begin{table*}[!t]
\centering
\caption{Fixed-reference absorption-capacity framework. The standard PPT
test and fully entangled fraction remain the underlying criteria. The
framework expresses their crossings along a fixed reference-mixing line as
capacities of the separable endpoint.}
\label{tab:prior-positioning}
\renewcommand{\arraystretch}{1.6}
\setlength{\tabcolsep}{3pt}
\hyphenpenalty=10000\exhyphenpenalty=10000
\begin{tabular}{p{0.26\textwidth}|p{0.32\textwidth}|p{0.37\textwidth}}
\hline\hline
\RaggedRight{}Setting
& \RaggedRight{}Existing viewpoint
& \RaggedRight{}This paper's contribution
\tabularnewline
\hline
\RaggedRight{}Werner / maximally mixed noise
& \RaggedRight{}A single isotropic special case,
$\lambda_*=1/3$.
& \RaggedRight{}Recovered as the maximally mixed point of the product
law, where $\EAC=1/2$ and $\lambda_*=1/3$.
\tabularnewline
\hline
\RaggedRight{}Bell-diagonal noise
& \RaggedRight{}Linear PPT and teleportation thresholds in the Bell
weights.
& \RaggedRight{}Recovered as a Bell-frame specialization of the
complex separable $X$ formula.
\tabularnewline
\hline
\RaggedRight{}Product noise $A\otimes B$
& \RaggedRight{}The Werner case gives the maximally mixed product
endpoint, but it does not determine the threshold for general product
endpoints.
& \RaggedRight{}Exact arbitrary-product formula
$\EAC=\sqrt{(1-\Tr A^2)(1-\Tr B^2)}$, giving
$\lambda_*=\EAC/(1+\EAC)$ from two local purities. The fidelity
absorption capacity $\EACF$ additionally depends on the product state's
orientation relative to the reference, and this determines $\lambda_F$.
\tabularnewline
\hline
\RaggedRight{}Pure-state reference
& \RaggedRight{}Requires recomputing where the partially transposed
mixture ceases to be positive.
& \RaggedRight{}For product noise, $\EAC$ rescales by $1/C_\psi$ for
every pure reference. For product-noise $\EACF$, we derive the corresponding
law for arbitrary pure references. For $X$-state noise, the rescaling is used
when the noise
state is $X$-shaped in the Schmidt basis of the reference.
\tabularnewline
\hline
\RaggedRight{}Separable $X$ noise under relaxation and dephasing
& \RaggedRight{}Without the closed forms, the PPT or teleportation
condition must be checked separately at each channel parameter.
& \RaggedRight{}Threshold trajectories are obtained by substituting the
channel-evolved $X$ entries into the closed forms for $\EAC$ and
$\EACF$.
\tabularnewline
\hline
\RaggedRight{}Teleportation threshold and entangled-but-unuseful interval
& \RaggedRight{}The fully entangled fraction is usually checked
separately. The interval where the state is entangled but not directly
useful for teleportation is not expressed directly in terms of the noise
state.
& \RaggedRight{}$\EACF$ gives $\lambda_F=\EACF/(1+\EACF)$. Together
with $\EAC$, it identifies the interval
$\lambda_*(\sigma)<\lambda\le\lambda_F(\sigma)$ where the Bell-mixed
state is entangled but still below the teleportation threshold.
\tabularnewline
\hline\hline
\end{tabular}
\end{table*}

\subsection{Using the framework}
\label{subsec:using-framework}
The theorems above are used from the noise state $\sigma$. Once the
capacity of $\sigma$ is known, the Bell weight at which
$\rho_\lambda$ becomes entangled, and the Bell weight at which it becomes
useful for teleportation, follow from the same M\"obius map
\eqref{eq:mobius}. This avoids solving a positivity or teleportation test
at every value of $\lambda$. The procedure is the same for a calibrated
noise model, for the closed-form product and $X$ families, and for a noise
state tracked along a channel: fix the reference, evaluate the capacity of
$\sigma$, and apply \eqref{eq:mobius}.

The \emph{reference} is the pure state whose projector is mixed into the
noise state. Unless stated otherwise, it is $|\PhiP\rangle$. Other Bell
references are handled by replacing the Bell frame. Pure references are
also allowed, but the dependence on the reference differs for the two
capacities. For product noise, $\EAC$ has the $1/C_\psi$ rescaling of
Eq.~\eqref{eq:EAC-pure-signal}. For product-noise $\EACF$, Bell
references use Eq.~\eqref{eq:EACF-product}, while arbitrary pure
references are evaluated in Appendix~\ref{subsec:supp-product-tele}. The
simple $1/C_\psi$ rescaling for product-noise $\EACF$ appears only in the
Schmidt-$z$ aligned case of Eq.~\eqref{eq:EACF-product-pure-aligned}. For
$X$-state noise, the pure-reference rescaling is used only when the noise
state is $X$-shaped in the Schmidt basis of the reference.

The input data depend on the noise family:
\begin{itemize}
\item \emph{Product noise.} The entanglement threshold needs only the two
local purities, through Eq.~\eqref{eq:impurity}. The teleportation
threshold also needs the Bell-frame alignment $m_\beta$ of
Eq.~\eqref{eq:EACF-product}.
\item \emph{Separable $X$ noise.} Read off the entries $(a,b,c,d,u,v)$ in
the convention of \eqref{eq:summary-Xstate}. Then use
Eqs.~\eqref{eq:EAC-X-PhiP}--\eqref{eq:EAC-X-PsiM} for $\EAC$ and
Eqs.~\eqref{eq:EAC-F-X}--\eqref{eq:EACF-X-PsiM} for $\EACF$.
\item \emph{General separable noise.} Use the definitions
\eqref{eq:EAC-def} and \eqref{eq:EACF-def}, or the variational
evaluations in Appendix~\ref{subsec:evaluating-capacities}: a
Rayleigh-quotient formula for $\EAC$ and a maximally entangled witness
formula for $\EACF$.
\end{itemize}

Decompositions are also useful when $\sigma$ itself is not in a
closed-form family. Suppose $\sigma=\sum_i p_i\sigma_i$, and suppose each
component has a feasible absorption amount $\eta_i$. Then the weighted
amount $\sum_i p_i\eta_i$ is feasible for $\sigma$:
$$
\sigma^{T_B}+\Bigl(\sum_i p_i\eta_i\Bigr)Q
=\sum_i p_i\bigl(\sigma_i^{T_B}+\eta_iQ\bigr)\succeq0.
$$
This identity is a certificate: it proves a lower bound on $\EAC(\sigma)$
even when the exact capacity is unknown. The same decomposition argument
applies to $\EACF$, since each maximally entangled witness condition is
affine in $(\sigma,\eta)$. Product or $X$ decompositions can therefore
certify lower bounds for both thresholds. Section~\ref{subsec:decomposition-certificates}
gives the corresponding concavity statements.

The two thresholds split the Bell-mixing line into three regimes:
$\rho_\lambda$ is separable for $\lambda\le\lambda_*$, entangled but below
the classical teleportation bound for $\lambda_*<\lambda\le\lambda_F$, and
teleportation-useful for $\lambda>\lambda_F$.

The product formulas make the workflow explicit on replacement channels.
Section~\ref{subsec:replacement-examples} gives a minimal demonstration
for a one-sided replacement channel, including the finite-step loss of
teleportation usefulness and separability for an initial pure entangled
state.

\section{Product law: entanglement and teleportation}
\label{sec:eval1}
On product noise
$\sigma=A\otimes B$, the joint two-qubit Bell-mixing entanglement
threshold $\lambda_*(\sigma)$ is fixed by the two single-qubit
marginal purities $\Tr A^2$ and $\Tr B^2$ alone. The required inputs are
only local marginal purities, obtainable from single-qubit marginal
tomography. No semidefinite program needs to be solved. The teleportation
threshold $\lambda_F$ is given in closed form on the same family. Unlike
$\lambda_*$ it also records the product state's orientation relative to the
reference.

\subsection{Entanglement absorption capacity}
\begin{lemma}[Bell-block inequality]
\label{lem:bell-block}
For any pair of single-qubit density matrices $A,B$,
\begin{equation}
A\otimes B^T+2\sqrt{\det A\,\det B}\,Q\succeq 0,
\label{eq:bellblock}
\end{equation}
with the sharp constant $2\sqrt{\det A\,\det B}$, the largest for which the
left-hand side stays positive semidefinite.
\end{lemma}
The standalone inequality is equally valid with $B$ in place of $B^T$, since
transposition preserves positivity and determinant. We keep $B^T$ because
this is the form produced by the partial transpose,
$(A\otimes B)^{T_B}=A\otimes B^T$, in the PPT test. Reading off the
supremum in \eqref{eq:EAC-def},
\begin{equation}
\EAC(A\otimes B)=2\sqrt{\det A\,\det B}.
\label{eq:EAC-product}
\end{equation}
Using
$2\det X=1-\Tr X^2$ for any single-qubit density matrix gives the
impurity form
\begin{equation}
\EAC(A\otimes B)=\sqrt{(1-\Tr A^2)(1-\Tr B^2)},
\label{eq:impurity}
\end{equation}
read off from the two local marginals alone. In Bloch coordinates
$A=(I+\vec a\!\cdot\!\boldsymbol\sigma)/2$,
$B=(I+\vec b\!\cdot\!\boldsymbol\sigma)/2$ this reads
\begin{equation}
\EAC(A\otimes B)
=\frac{\sqrt{(1-|\vec a|^2)(1-|\vec b|^2)}}{2}.
\label{eq:bloch-form}
\end{equation}
The product entanglement absorption capacity therefore depends only on the local
Bloch radii: it vanishes when either marginal is pure and is maximal
only when both marginals are maximally mixed. Although
Eq.~\eqref{eq:impurity} is derived in the $\PhiP$ frame, it holds
verbatim for every maximally entangled reference: any such reference is
$(W\otimes I)|\PhiP\rangle$ for a single-qubit unitary $W$, and
conjugating the mixing line by $W^\dagger\otimes I$ sends $A\otimes B$ to
$(W^\dagger A W)\otimes B$, a product noise with the same local
determinants and partial-transpose spectrum, so $\EAC$ is
reference-independent on the product family. The four Bell frames are the
special case $W\in\{I,\sigma_x,\sigma_y,\sigma_z\}$.

\emph{Robustness reading.} Let $R(\PhiP\,\|\,\sigma)$ denote the
Vidal--Tarrach relative robustness: the least $t\ge0$ for which
$(\PhiP+t\sigma)/(1+t)$ is separable~\cite{VidalTarrach1999}. The
robustness parameter $t$ and the absorption parameter $\eta$ are
reciprocal on the same mixing line, since
$(\PhiP+t\sigma)/(1+t)=(\sigma+\eta\PhiP)/(1+\eta)$ when
$\eta=1/t$. The same line gives $\EAC(\sigma)=1/R(\PhiP\,\|\,\sigma)$, with
$\EAC=0$ corresponding to $R=\infty$. For product noise,
Eq.~\eqref{eq:impurity} gives the closed form
\begin{equation}
R(\PhiP\,\|\,A\otimes B)
=\frac{1}{2\sqrt{\det A\,\det B}}
=\frac{1}{\sqrt{(1-\Tr A^2)(1-\Tr B^2)}}.
\label{eq:product-robustness}
\end{equation}
At the maximally mixed product endpoint $A=B=I_2/2$, the formula gives
$R=2$. Away from that symmetric point, it resolves the full dependence on
the two local marginals.

Substituting \eqref{eq:impurity} into the threshold identity
\eqref{eq:mobius} gives the corresponding Bell-mixing threshold,
\begin{equation}
\lambda_*(A\otimes B)
=
\frac{\sqrt{(1-\Tr A^2)(1-\Tr B^2)}}
{1+\sqrt{(1-\Tr A^2)(1-\Tr B^2)}}
\,\le\,\tfrac13,
\label{eq:product-threshold}
\end{equation}
with equality if and only if $A=B=I_2/2$. Werner's classical
$\lambda_*=1/3$ is the unique maximum on the product family, attained at
maximally mixed marginals. For every other product endpoint,
$\lambda_*(A\otimes B)<1/3$, so the Bell mixture becomes entangled at a
smaller Bell weight.
Figure~\ref{fig:product-threshold} shows both
forms: the map from the capacity interval
$\EAC\in[0,1/2]$ to the threshold interval $\lambda_*\in[0,1/3]$, and
the same surface in local Bloch coordinates.

The maximum at maximally mixed marginals is specific to the product
family, not the whole separable set. Concavity of $\EAC$ gives
lower-bound certificates from decompositions of $\sigma$, but it does not
make the product-surface maximum global. To see the difference, suppose a
Bell-mixing point is separable and write
\begin{equation*}
\rho_\lambda=\sum_k p_k\,
|a_k\rangle\!\langle a_k|\otimes |b_k\rangle\!\langle b_k| .
\end{equation*}
Then
\begin{equation*}
\langle\PhiP|\rho_\lambda|\PhiP\rangle
=\sum_k p_k\,|\langle\PhiP|a_k\otimes b_k\rangle|^2
\le \sum_k p_k\,\frac12
=\frac12 .
\end{equation*}
For the Bell-mixing line,
$\langle\PhiP|\rho_\lambda|\PhiP\rangle=\lambda+(1-\lambda)\langle\PhiP|\sigma|\PhiP\rangle$,
this gives
\begin{equation}
\lambda\le
\frac{\tfrac12-\langle\PhiP|\sigma|\PhiP\rangle}
{1-\langle\PhiP|\sigma|\PhiP\rangle}
\le \tfrac12,
\qquad
\EAC(\sigma)\le 1-2\langle\PhiP|\sigma|\PhiP\rangle\le1 .
\label{eq:global-sep-EAC-bound}
\end{equation}
The capacity bound follows from the first inequality by the inverse
M\"obius map $\EAC=\lambda_*/(1-\lambda_*)$. The bound $\EAC\le1$ is
saturated, but not at the maximally mixed state, which only reaches
$\tfrac12$. For example,
$\sigma=\tfrac12(|01\rangle\!\langle01|+|10\rangle\!\langle10|)$
has $\EAC(\sigma;\PhiP)=1$ and $\lambda_*=1/2$. This shows that
$I_4/4$ is the unique maximizer only within the product-noise family,
where the independent local-noise constraint enforces the sharper
cap $\EAC\le1/2$.

\begin{figure*}[!t]
\centering
\includegraphics[width=0.92\textwidth]{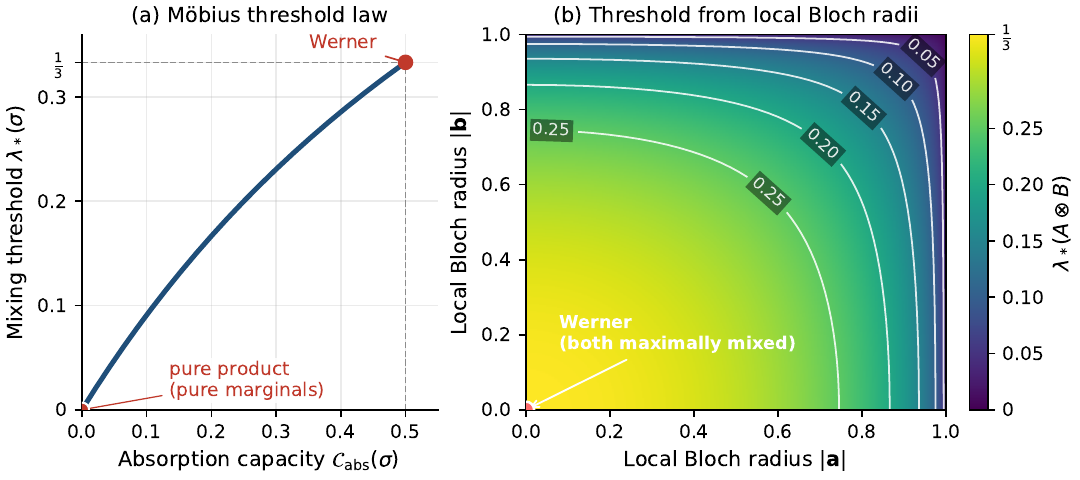}
\caption{Two complementary views of the product-law threshold (absorption
capacity $\EAC$ defined in Section~\ref{sec:eac}, product-noise closed form
in Eq.~\eqref{eq:impurity}).  (a) The
map $\lambda_*=\EAC/(1+\EAC)$ sends the product-noise capacity
interval $\EAC\in[0,1/2]$ to the threshold interval
$\lambda_*\in[0,1/3]$, with the
Werner point $\lambda_*=1/3$ reached only at $\EAC=1/2$.  (b) In Bloch
coordinates the same threshold depends only on the local radii
$|\vec a|$ and $|\vec b|$ of the two marginals.  The threshold
vanishes when either marginal becomes pure and is maximal only at the
double maximally mixed point $(|\vec a|,|\vec b|)=(0,0)$.}
\label{fig:product-threshold}
\end{figure*}

\textit{Proof sketch of Lemma~\ref{lem:bell-block}.}
If $\det A=0$ or $\det B=0$, the result is immediate. Otherwise:
the swap operator
$\Swap=2Q=\sum_{i,j=0}^{1}|ij\rangle\!\langle ji|$
on $\mathbb C^2\!\otimes\!\mathbb C^2$ commutes with $U\otimes U$ for
every unitary $U$, so
choosing $U$ with $UAU^\dagger=\operatorname{diag}(a,d)$
reduces the problem to $A=\operatorname{diag}(a,d)$ with $B^T$
replaced by $U B^T U^\dagger$, a positive Hermitian matrix with the
same determinant. With $t_*:=\sqrt{ad\det B}=\sqrt{\det A\det B}$
and the transformed factor $UB^TU^\dagger$ written as
$
\begin{psmallmatrix}
p&\overline q\\ q&r
\end{psmallmatrix}
$,
the matrix $M:=A\otimes B^T+2t_*Q$, expressed in the permuted basis
$\{|00\rangle,|11\rangle,|01\rangle,|10\rangle\}$, takes a block form
whose upper-left block is diagonal positive.  A direct Schur-complement
computation (see Appendix~\ref{sec:supp-proof}) shows that the Schur
complement has positive trace and identically vanishing determinant
at $t_*$, hence is positive semidefinite.  Therefore $M\succeq0$.
The full pointwise sharpness argument is given in Appendix~\ref{sec:supp-proof}: an
explicit factorization of the Schur-determinant numerator shows that
the analogous matrix $M(t):=A\otimes B^T+2tQ$ fails $\succeq 0$ for
every $t>t_*$, so the constant in Eq.~\eqref{eq:bellblock} is
pointwise sharp in $(A,B)$.
The Werner point $A=B=I_2/2$ illustrates saturation: $I_4/4+\eta Q$
has minimum eigenvalue $1/4-\eta/2$, vanishing at
$\eta=1/2=2\sqrt{\det A\,\det B}$.

\paragraph*{Arbitrary reference state.}
An entanglement-distribution source may prepare a
pure-state reference $|\psi\rangle$ rather than an ideal
Bell reference state. The operational question is then how the entanglement
threshold changes when the reference is an arbitrary pure state. On product
noise the answer is again
compact: the dependence on the pure-state reference enters only through
the scalar $C_\psi$, the concurrence of $|\psi\rangle$~\cite{Wootters1998},
and the product-state law
\eqref{eq:EAC-product} generalizes from the maximally entangled Bell
reference state $|\PhiP\rangle$ to any two-qubit pure-state reference
$|\psi\rangle$. We denote the corresponding rank-one projector by
$\psi:=|\psi\rangle\!\langle\psi|$. Conjugating the entire mixing line
$\rho_\lambda=\lambda\psi+(1-\lambda)\sigma$
by a pair of local unitaries $U\otimes V$ leaves the entanglement
threshold (and hence $\EAC(\sigma;\psi)$) invariant, because local
unitaries preserve the partial-transpose spectrum. Choosing
$U\otimes V$ to send $|\psi\rangle$ to its Schmidt form simultaneously
sends a product noise state $\sigma=A\otimes B$ to
$(UAU^\dagger)\otimes(VBV^\dagger)$, with the same local determinants.
The product-noise capacity depends only on these determinants, so we may
assume Schmidt form
\begin{equation}
|\psi\rangle = \alpha\,|00\rangle + \beta\,|11\rangle,
\quad \alpha,\beta\ge 0,\quad \alpha^2+\beta^2=1,
\label{eq:psi-schmidt}
\end{equation}
without loss of generality on product noise. The concurrence of
$|\psi\rangle$ is then
$C_\psi = 2\alpha\beta\in[0,1]$~\cite{Wootters1998}.
In this form,
$|\psi\rangle = \sqrt2\,(I\otimes R)|\PhiP\rangle$ with
$R=\mathrm{diag}(\alpha,\beta)$.

\paragraph*{Filtered Bell-block inequality.}
Let $Q_\psi:=(|\psi\rangle\!\langle\psi|)^{T_B}$ denote the partial
transpose of the projector onto the pure-state reference. For this
reference, the absorption capacity is
$$
\EAC(\sigma;\psi):=
\sup\{\eta\ge0:\ \sigma^{T_B}+\eta Q_\psi\succeq0\}.
$$
The Bell-block inequality of Lemma~\ref{lem:bell-block} extends to
arbitrary pure-state references as follows.
\begin{lemma}[Filtered Bell-block inequality]
\label{lem:filtered-bell-block}
For any two-qubit pure-state reference $|\psi\rangle$ with concurrence
$C_\psi>0$ and any single-qubit density matrices $A,B$,
\begin{equation}
A\otimes B^T+\frac{2\sqrt{\det A\,\det B}}{C_\psi}\,Q_\psi
\succeq 0,
\label{eq:filtered-bellblock}
\end{equation}
with the sharp constant $2\sqrt{\det A\,\det B}/C_\psi$, the largest for
which the left-hand side stays positive semidefinite.
\end{lemma}
The proof is a short congruence calculation, given in
Appendix~\ref{subsec:supp-filtered-block}. Reading off this supremum,
\begin{equation}
\EAC(A\otimes B;\psi)
=\frac{2\sqrt{\det A\,\det B}}{C_\psi}.
\label{eq:EAC-pure-signal}
\end{equation}
Using $2\det X=1-\Tr X^2$ gives the impurity form
\begin{equation}
\EAC(A\otimes B;\psi)
=\frac{\sqrt{(1-\Tr A^2)(1-\Tr B^2)}}{C_\psi}.
\label{eq:EAC-pure-impurity}
\end{equation}
For $C_\psi=1$ (Bell reference state) this reduces to
Eq.~\eqref{eq:impurity}. Lemma~\ref{lem:filtered-bell-block} assumes
$C_\psi>0$. When $C_\psi=0$, the reference is a product state, so the
product-noise mixture stays separable for all $\lambda$ and $\EAC$
is infinite.

\paragraph*{Remark: only the magnitudes enter.}
Equation~\eqref{eq:EAC-pure-signal} depends only on magnitudes, the
concurrence $C_\psi$ and the determinants $\det A,\det B$, and not on the
relative orientation of the vectors involved: the angles among $\vec a$,
$\vec b$, and the Schmidt frame of $|\psi\rangle$ play no role. This goes
beyond local-unitary invariance, which would only fix $\EAC$ under a
\emph{joint} rotation of $A$, $B$, and $|\psi\rangle$: here rotating $A$
alone, at fixed $\det A$, moves $\vec a$ relative to $\vec b$ and the
Schmidt frame yet leaves \eqref{eq:filtered-bellblock} saturated at the
same critical constant.

\paragraph*{Two-parameter threshold law.}
Substituting \eqref{eq:EAC-pure-signal} into the threshold identity
\eqref{eq:mobius} gives the Bell-mixing threshold,
\begin{equation}
\lambda_*(A\otimes B;\psi)
=\frac{\sqrt{(1-\Tr A^2)(1-\Tr B^2)}}
{C_\psi+\sqrt{(1-\Tr A^2)(1-\Tr B^2)}}.
\label{eq:product-pure-threshold}
\end{equation}
The noise enters only through its local impurities $1-\Tr A^2$ and $1-\Tr B^2$, the reference only through $C_\psi$. Two monotonicities
follow immediately, in opposite directions. For fixed product noise,
$\lambda_*$ strictly \emph{decreases} in $C_\psi$: a more entangled
reference makes $\rho_\lambda$ entangled at a smaller Bell weight. For a
fixed pure-state reference $|\psi\rangle$, $\lambda_*$ strictly \emph{increases}
in $\sqrt{(1-\Tr A^2)(1-\Tr B^2)}$: more mixed local marginals keep
$\rho_\lambda$ separable up to a larger Bell weight. The impurity factor $g:=\sqrt{(1-\Tr A^2)(1-\Tr B^2)}$ is at most
$\tfrac12$, and $\EAC=g/C_\psi$. It is largest, $g=\tfrac12$, at the
maximally mixed marginals $A=B=I_2/2$. Therefore, for fixed reference
concurrence, the product-noise threshold satisfies
\begin{equation}
\lambda_*(A\otimes B;\psi)\le\frac{1}{2C_\psi+1},
\label{eq:product-pure-ceiling}
\end{equation}
with equality at $A=B=I_2/2$. The Werner value $\lambda_*=1/3$ is the
maximally entangled case $C_\psi=1$. As $C_\psi$ decreases, the upper
bound increases and tends to $1$ as $C_\psi\to0$. At $C_\psi=0$ the
reference is a product state, so the product-noise mixing line has no
entanglement crossing.

\paragraph*{Channel corollary.}
Equation~\eqref{eq:bloch-form} shows that on product noise
$\EAC(A\otimes B)$ depends only on the two marginal Bloch radii. Decreasing
either radius increases the corresponding local impurity, so $\EAC$ and
hence $\lambda_*=\EAC/(1+\EAC)$ cannot decrease. If
$A=(I+\vec a\cdot\boldsymbol\sigma)/2$ and
$\mathcal N(A)=(I+\vec a\,'\cdot\boldsymbol\sigma)/2$, then a local unital
qubit channel satisfies $|\vec a\,'|\le|\vec a|$. Therefore local unital
noise cannot lower $\EAC$ or $\lambda_*$ on the product family. This
monotonicity is not a property of arbitrary local channels. Nonunital
channels do not have to contract every Bloch radius, so they can either
increase or decrease the product capacity depending on the input
marginals. The precise statement and an explicit amplitude-damping
counterexample are given in
Theorem~\ref{thm:channel-monotonicity} of
Appendix~\ref{subsec:supp-monotonicity}.

\subsection{Fidelity absorption capacity}
The fidelity absorption capacity also has a closed form on the full
product family, but it uses more information than the local purities.
For a Bell reference $\beta\in\{\Phi^+,\Phi^-,\Psi^+,\Psi^-\}$, let
$D_\beta$ be the sign matrix of its correlation tensor,
\begin{equation}
\begin{aligned}
D_{\Phi^+}&=\operatorname{diag}(1,-1,1), &\quad
D_{\Phi^-}&=\operatorname{diag}(-1,1,1),\\
D_{\Psi^+}&=\operatorname{diag}(1,1,-1), &\quad
D_{\Psi^-}&=\operatorname{diag}(-1,-1,-1),
\end{aligned}
\label{eq:Dbeta}
\end{equation}
where $A=(I+\vec a\cdot\boldsymbol\sigma)/2$ and
$B=(I+\vec b\cdot\boldsymbol\sigma)/2$. Define
$h:=|\vec a|\,|\vec b|$ and the Bell-frame alignment
$m_\beta:=\vec a^{\,T}D_\beta\vec b$. Then
\begin{equation}
\EACF(A\otimes B;\beta)=\frac{\sqrt{1-h^2+m_\beta^2}-m_\beta}{2},
\label{eq:EACF-product}
\end{equation}
with $\lambda_F=\EACF/(1+\EACF)$.

This is the product-family distinction between the two capacities:
the entanglement absorption capacity in Eq.~\eqref{eq:impurity} depends
only on the two Bloch radii and is the same for every maximally entangled reference,
whereas $\EACF$
depends also on the signed alignment $m_\beta$. The difference is already
visible in two $|\PhiP\rangle$-frame orientations. If
$\vec a=|\vec a|\hat z$ and $\vec b=|\vec b|\hat z$, then
$$
\lambda_F=\frac{1-|\vec a|\,|\vec b|}{3-|\vec a|\,|\vec b|}.
$$
In this aligned case, $\lambda_F=\lambda_*$ exactly on the equal-radius
diagonal $|\vec a|=|\vec b|$. If instead
$\vec a=|\vec a|\hat z$ and $\vec b=|\vec b|\hat x$, then
$$
\lambda_F=\frac{\sqrt{1-|\vec a|^2|\vec b|^2}}
{2+\sqrt{1-|\vec a|^2|\vec b|^2}}.
$$
In this perpendicular case, $\lambda_F=\lambda_*$ only at
$|\vec a|=|\vec b|=0$ and at the pure-product boundary point
$|\vec a|=|\vec b|=1$. The equality set therefore depends on orientation
relative to the Bell frame. Away from the equality set, $\lambda_F>\lambda_*$,
leaving a Bell-weight interval in which the Bell-mixed state is entangled
but still below the teleportation threshold. Appendix~\ref{subsec:supp-product-tele}
derives the product-family formula \eqref{eq:EACF-product}, and
Appendix~\ref{subsec:supp-subfamily-checks} checks the two orientations
above.

\paragraph*{Arbitrary reference state.}
For $\EAC$, changing the Bell reference to a pure reference of concurrence
$C_\psi$ only introduces the factor $1/C_\psi$ in
Eq.~\eqref{eq:EAC-pure-signal}. The fidelity absorption capacity $\EACF$ is different:
it also depends on the directions of the product Bloch vectors relative to
the Schmidt basis of the reference state. One aligned case remains simple.
Let $|\psi\rangle=\alpha|00\rangle+\beta|11\rangle$ be in Schmidt form in
the computational basis, with $C_\psi=2\alpha\beta>0$, and take both product
Bloch vectors along the same $z$ axis,
$\vec a=|\vec a|\hat z$ and $\vec b=|\vec b|\hat z$. Then
\begin{equation}
\EACF(A\otimes B;\psi)=\frac{1-|\vec a|\,|\vec b|}{2C_\psi},
\qquad
\lambda_F=\frac{1-|\vec a|\,|\vec b|}
{2C_\psi+1-|\vec a|\,|\vec b|}.
\label{eq:EACF-product-pure-aligned}
\end{equation}
In this aligned case, $\EACF$ has the same $1/C_\psi$ rescaling as
$\EAC$. For general product directions the crossing depends on the full
relative orientation.
Appendix~\ref{subsec:supp-product-tele} gives the product-noise
teleportation calculation in two stages: first the closed Bell-reference
formula \eqref{eq:EACF-product}, and then its arbitrary pure-reference
extension, where the threshold is obtained from a quartic equation in
$\eta$ and the positive-determinant branch of the Horodecki formula gives
no crossing.

\subsection{Replacement-channel examples}
\label{subsec:replacement-examples}
Replacement channels give simple threshold examples because their endpoints
are product states. Consider the one-sided replacement channel
$$
\mathcal R_{\lambda,\tau}(\omega)=\lambda\omega+(1-\lambda)\Tr(\omega)\tau,
\qquad
\tau=\frac{I+\vec r\cdot\boldsymbol\sigma}{2},
$$
acting on the second qubit:
\begin{equation}
(I\otimes\mathcal R_{\lambda,\tau})(\rho)
=
\lambda\rho+(1-\lambda)\Tr_2(\rho)\otimes\tau .
\label{eq:one-sided-replacement}
\end{equation}
Let the input be a pure entangled state
$\psi=|\psi\rangle\!\langle\psi|$ and set $A=\Tr_2\psi$. The first
marginal is unchanged by the channel. Hence $\Tr_2\rho_n=A$ at every
step, and the recursion closes on the two states $\psi$ and
$A\otimes\tau$:
\begin{equation}
\rho_n
=
\lambda^n\psi+(1-\lambda^n)A\otimes\tau .
\label{eq:one-sided-iterate}
\end{equation}
The limiting endpoint $A\otimes\tau$ has, by
\eqref{eq:EAC-pure-signal},
$$
\EAC(A\otimes\tau;\psi)
=
\frac{2\sqrt{\det A\,\det\tau}}{C_\psi}
=
\sqrt{\det\tau}
=
\frac{\sqrt{1-|\vec r|^2}}{2},
\qquad
C_\psi=2\sqrt{\det A}.
$$
The dependence on the input concurrence cancels. For mixed $\tau$, the
entanglement threshold and the first separable iterate are
\begin{equation}
\lambda_*
=
\frac{\sqrt{1-|\vec r|^2}}{2+\sqrt{1-|\vec r|^2}},
\qquad
n_{\rm ESD}
=
\left\lceil
\frac{\log\lambda_*}{\log\lambda}
\right\rceil .
\label{eq:one-sided-esd-time}
\end{equation}
If $\tau$ is pure, then $\det\tau=0$ and $\lambda_*=0$. For
$0<\lambda<1$, the state remains entangled at every finite iterate and
becomes separable only in the limit $n\to\infty$.

For the Bell input $\psi=\PhiP$ the marginal is $A=I_2/2$, and a single
use gives the channel's Choi state,
$$
(I\otimes\mathcal R_{\lambda,\tau})(\PhiP)
=
\lambda\PhiP +(1-\lambda)\frac{I_2}{2}\otimes\tau .
$$
For this noise state,
$$
\EACF\!\left(\frac{I_2}{2}\otimes\tau;\PhiP\right)=\frac12,
$$
independent of $\tau$, so the teleportation threshold is
$\lambda_F=1/3$. Since $\lambda_*\le\lambda_F$, the Choi trajectory
crosses the teleportation threshold no later than the entanglement
threshold:
$$
n_F=
\left\lceil\frac{\log(1/3)}{\log\lambda}\right\rceil
\le n_{\rm ESD}.
$$

The maximally mixed replacement $\tau=I_2/2$ is the one-sided
depolarizing channel $\mathcal D_p(\omega)=(1-p)\omega+pI_2/2$, with
$\lambda=1-p$. For the Bell input,
$$
(I\otimes\mathcal D_p)(\PhiP)
=
(1-p)\PhiP + p\,\frac{I_4}{4}.
$$
Here
$\EAC(I_4/4)=\EACF(I_4/4)=1/2$, so
$\lambda_*=\lambda_F=1/3$. The one-sided depolarizing channel is
entanglement breaking for $p\ge2/3$, and its Choi state beats the
classical teleportation bound precisely for $p<2/3$. Depolarizing both
qubits gives Bell weight $(1-p)^2$, so the two-sided crossing occurs
at $p=1-1/\sqrt3$.

A second replacement example overwrites both parties. Replacing the pair
by $\tau_A\otimes\tau_B$, with local Bloch vectors
$\vec r_A,\vec r_B$, gives
$$
\rho_\lambda=\lambda\PhiP +(1-\lambda)\tau_A\otimes\tau_B .
$$
The noise state $\tau_A\otimes\tau_B$ has
$$
\EAC(\tau_A\otimes\tau_B)
=
\frac{\sqrt{(1-|\vec r_A|^2)(1-|\vec r_B|^2)}}{2},
$$
and
$$
\EACF(\tau_A\otimes\tau_B)
=
\frac{\sqrt{1-h^2+m^2}-m}{2},
\qquad
h=|\vec r_A|\,|\vec r_B|,
\qquad
m=\vec r_A^{\,T}\operatorname{diag}(1,-1,1)\vec r_B .
$$
From these, $\lambda_*=\EAC/(1+\EAC)$ and $\lambda_F=\EACF/(1+\EACF)$
give the two thresholds.

More generally, if the noise state is a separable mixture
$$
\sigma=\sum_i p_i A_i\otimes B_i,
$$
then concavity~(Eqs.~\eqref{eq:EAC-concavity},~\eqref{eq:EACF-concavity})
gives lower bounds on both capacities:
$$
L_*=\sum_i p_i\EAC(A_i\otimes B_i),
\qquad
L_F=\sum_i p_i\EACF(A_i\otimes B_i).
$$
These show that $\rho_\lambda$ is separable for $\lambda\le L_*/(1+L_*)$
and not teleportation-useful for $\lambda\le L_F/(1+L_F)$.
For example, if every term has fixed local purities
$\Tr A_i^2=\pi_A$ and $\Tr B_i^2=\pi_B$, then
$$
\EAC(\sigma)\ge
g:=\sqrt{(1-\pi_A)(1-\pi_B)},
\qquad
\lambda_*(\sigma)\ge\frac{g}{1+g},
$$
regardless of the orientations and weights of the components. If the
noise state $\sigma$ is an $X$ state, such as a Bell-diagonal separable
mixture, the $X$-state formulas apply similarly.

\section{\texorpdfstring{$X$-state}{X-state} law}
\label{sec:eval2}
We now consider separable noise in the two-qubit $X$-state family defined
in \eqref{eq:summary-Xstate}. These states keep the computational-basis
block structure and may carry the two antidiagonal coherences $u$ and $v$.
This is the next case beyond product noise where the Bell-mixing thresholds
can still be explicit.

The $X$ family contains product states with computational-basis diagonal
marginals and Bell-diagonal states. Both $\EAC$ and $\EACF$ admit closed
forms on the full separable $X$ sector, and the Bell-diagonal and
diagonal-product laws follow by substitution. We use the $\Phi/\Psi$ block
convention of \eqref{eq:summary-Xstate}: the $\Phi$ block carries $u$ and
the $\Psi$ block carries $v$. A separable $X$ state satisfies the four
inequalities $ad\ge|u|^2$, $bc\ge|v|^2$, $ad\ge|v|^2$, and
$bc\ge|u|^2$, where the first two are the density-matrix conditions in
\eqref{eq:summary-Xstate} and the last two are the PPT conditions.

\subsection{Entanglement absorption capacity}
Since $\Phi^+$ is itself an $X$ state, the Bell-mixing point
$$
\rho_\eta=\frac{\sigma_X+\eta\Phi^+}{1+\eta}
$$
remains in the $X$ family for every $\eta\ge0$, regardless of whether it
is separable or entangled.
For the $|\PhiP\rangle$ reference, the SDP becomes two $2\times2$
positivity tests after partial transposition. The partial transpose swaps
the two antidiagonal coherences of an $X$ state: $u$, originally in the
$(|00\rangle,|11\rangle)$ entry of $\sigma_X$, moves to the
$(|01\rangle,|10\rangle)$ entry of $\sigma_X^{T_B}$, while $v$ moves in
the opposite direction. In the computational basis,
$$
\sigma_X^{T_B}=
\begin{pmatrix}
a & 0 & 0 & v\\
0 & b & u & 0\\
0 & u^* & c & 0\\
v^* & 0 & 0 & d
\end{pmatrix}.
$$
The partial transpose of the $|\PhiP\rangle$ projector is
$$
Q=(|\PhiP\rangle\!\langle\PhiP|)^{T_B}
=\tfrac{1}{2}\bigl(|00\rangle\!\langle 00|
+|01\rangle\!\langle 10|+|10\rangle\!\langle 01|+|11\rangle\!\langle 11|\bigr)
=\frac12
\begin{pmatrix}
1&0&0&0\\
0&0&1&0\\
0&1&0&0\\
0&0&0&1
\end{pmatrix}
$$
in the computational basis. Adding $\eta Q$ to $\sigma_X^{T_B}$ and
reordering to the basis order
$|00\rangle,|11\rangle,|01\rangle,|10\rangle$ gives
\begin{equation}
\begin{aligned}
\sigma_X^{T_B}+\eta Q
&\cong
\begin{pmatrix}
a+\tfrac{\eta}{2} & v & 0 & 0\\
v^* & d+\tfrac{\eta}{2} & 0 & 0\\
0 & 0 & b & u+\tfrac{\eta}{2}\\
0 & 0 & u^*+\tfrac{\eta}{2} & c
\end{pmatrix}
\\[2pt]
&=\underbrace{\begin{pmatrix}
a+\tfrac{\eta}{2} & v\\
v^* & d+\tfrac{\eta}{2}
\end{pmatrix}}_{\Phi\text{ block }(00/11)}
\oplus
\underbrace{\begin{pmatrix}
b & u+\tfrac{\eta}{2}\\
u^*+\tfrac{\eta}{2} & c
\end{pmatrix}}_{\Psi\text{ block }(01/10)},
\end{aligned}
\label{eq:X-blocks}
\end{equation}
where $\oplus$ denotes block-diagonal concatenation.
Although the reference added before partial transposition is
$\Phi^+$, the negative eigenspace of $Q=(\Phi^+)^{T_B}=S/2$ lies in the
$\Psi$ block: $Q|\Psi^-\rangle=-\frac12|\Psi^-\rangle$.
The SDP \eqref{eq:EAC-def} is therefore a one-variable positivity
condition on the $\Phi$ and $\Psi$ blocks in \eqref{eq:X-blocks}. At $\eta=0$ both blocks are
positive semidefinite by the PPT conditions. For $\eta\ge 0$, the
$\Phi$ block remains positive: its diagonal entries are nonnegative, and
its determinant
$$
ad-|v|^2+\frac{a+d}{2}\eta+\frac{\eta^2}{4}
$$
is nondecreasing from the nonnegative value $ad-|v|^2$. Hence the only
remaining constraint comes from the $\Psi$ block, whose determinant
condition is
$$
bc-|u+\eta/2|^2
=bc-\left(\operatorname{Re}u+\frac{\eta}{2}\right)^2
-\bigl(\operatorname{Im}u\bigr)^2
\ge 0.
$$
Solving this inequality for the largest admissible $\eta$ gives the
entanglement absorption capacity in the $|\PhiP\rangle$ frame:
\begin{equation}
\EAC(\sigma_X;\Phi^+)
=-2\,\operatorname{Re} u+2\sqrt{bc-(\operatorname{Im} u)^2}.
\label{eq:EAC-X-PhiP}
\end{equation}
The corresponding Bell-mixing threshold follows immediately from
\eqref{eq:mobius}:
\begin{equation}
\lambda_*(\sigma_X;\Phi^+)
=
\frac{-2\,\operatorname{Re} u+2\sqrt{bc-(\operatorname{Im} u)^2}}
{1-2\,\operatorname{Re} u+2\sqrt{bc-(\operatorname{Im} u)^2}}.
\label{eq:X-threshold-PhiP}
\end{equation}
The remaining Bell frames follow from the one-qubit Pauli maps that send
$\Phi^+$ to $\Phi^-$, $\Psi^+$, and $\Psi^-$. These maps preserve the
$X$ form and only flip the relevant coherence or exchange the two blocks,
giving
\begin{align}
\EAC(\sigma_X;\Phi^-)
&=
2\,\operatorname{Re} u+2\sqrt{bc-\bigl(\operatorname{Im} u\bigr)^2},
\label{eq:EAC-X-PhiM}
\\
\EAC(\sigma_X;\Psi^+)
&=
-2\,\operatorname{Re} v+2\sqrt{ad-\bigl(\operatorname{Im} v\bigr)^2},
\label{eq:EAC-X-PsiP}
\\
\EAC(\sigma_X;\Psi^-)
&=
2\,\operatorname{Re} v+2\sqrt{ad-\bigl(\operatorname{Im} v\bigr)^2}.
\label{eq:EAC-X-PsiM}
\end{align}
In the $|\PhiP\rangle$ frame, the entanglement absorption capacity depends only on the
coherence $u$ and the two populations $b,c$ in the
$|01\rangle,|10\rangle$ positions. If $u$ is real, Eq.~\eqref{eq:EAC-X-PhiP}
reduces to $\EAC(\sigma_X;\Phi^+)=2(\sqrt{bc}-u)$.

Equations \eqref{eq:EAC-X-PhiP}--\eqref{eq:EAC-X-PsiM} give the
absorption capacity for any separable two-qubit $X$ state in all four Bell
frames. For each Bell frame $\beta$, the Bell-mixing entanglement threshold
is then
$$
\lambda_*(\sigma_X;\beta)
=\frac{\EAC(\sigma_X;\beta)}{1+\EAC(\sigma_X;\beta)}.
$$
Partial transposition swaps the coherences carried by the $\Phi$ and $\Psi$
blocks. Therefore the $\Phi$-frame formula compares $u$ with the opposite
block-population product $bc$, while the $\Psi$-frame formula compares $v$
with $ad$. For real $X$ states, the same two block comparisons also appear
in the usual $X$-state concurrence on the entangled side, as discussed in
Section~\ref{subsec:optimized-margins}.

\paragraph*{Arbitrary reference state.}
On the $X$ sector, the simple $1/C_\psi$ rescaling has an aligned
Schmidt-basis form. In this case,
$|\psi\rangle=\alpha|00\rangle+\beta|11\rangle$ and the noise state is
$X$-shaped in the same computational basis
$|00\rangle,|01\rangle,|10\rangle,|11\rangle$. Product noise is
basis-independent for this law: local unitaries preserve tensor-product
form, and the capacity depends only on local determinants. The $X$ case
is different because the block structure is basis-dependent. If the
Schmidt basis of the reference is not the basis in which the noise state
has $X$ form, the matrix $\sigma_X^{T_B}+\eta Q_\psi$ is generally not
block diagonal in the $\Phi/\Psi$ decomposition. The positivity test then
does not reduce to the two $2\times2$ blocks used below.

In the aligned case, the partial-transposed projector is
$$
Q_\psi
=\alpha^2|00\rangle\!\langle 00|
+\alpha\beta\bigl(|01\rangle\!\langle 10|+|10\rangle\!\langle 01|\bigr)
+\beta^2|11\rangle\!\langle 11|,
$$
so, after reordering to the basis
$|00\rangle,|11\rangle,|01\rangle,|10\rangle$,
$$
\sigma_X^{T_B}+\eta Q_\psi
\cong
\underbrace{\begin{pmatrix}
a+\eta\alpha^2 & v\\
v^* & d+\eta\beta^2
\end{pmatrix}}_{\Phi\text{ block }(00/11)}
\oplus
\underbrace{\begin{pmatrix}
b & u+\eta\alpha\beta\\
u^*+\eta\alpha\beta & c
\end{pmatrix}}_{\Psi\text{ block }(01/10)} .
$$
The $\Phi$ block stays positive for all $\eta\ge0$. Positivity of the
$\Psi$ block is equivalent to $|u+\eta\alpha\beta|^2\le bc$. Solving for the largest
admissible $\eta$ proves the aligned clause for pure-state references in
Theorem~\ref{thm:xstate}: for
$C_\psi=2\alpha\beta>0$,
\begin{align}
\EAC(\sigma_X;\psi)
&=\frac{-2\,\operatorname{Re}u+2\sqrt{bc-(\operatorname{Im}u)^2}}{C_\psi}
\nonumber\\
&=\frac{\EAC(\sigma_X;\Phi^+)}{C_\psi},
\label{eq:xstate-pure-signal}
\end{align}
with the corresponding Bell-mixing entanglement threshold
$\lambda_*(\sigma_X;\psi)$ following from
Theorem~\ref{thm:mobius}.

\emph{Subsumed sectors.} The $X$-state law contains two familiar families
as special cases. For a product state with computational-basis diagonal
marginals, $\sigma_X=A\otimes B$ has $u=v=0$ and
$bc=\det A\,\det B$. Equation~\eqref{eq:EAC-X-PhiP} becomes
$\EAC(\sigma_X;\Phi^+)=2\sqrt{bc}=2\sqrt{\det A\,\det B}$, which is the
diagonal-marginal case of the product law \eqref{eq:EAC-product}. Product
states with local coherences need not be $X$ states, so the arbitrary
product formula remains the separate result of Section~\ref{sec:eval1}.
For the Bell-diagonal separable family, with Bell weights
$p_{\Phi^+},p_{\Phi^-},p_{\Psi^+},p_{\Psi^-}$, the $X$ entries are
$a=d=(p_{\Phi^+}+p_{\Phi^-})/2$,
$b=c=(p_{\Psi^+}+p_{\Psi^-})/2$,
$u=(p_{\Phi^+}-p_{\Phi^-})/2$, and
$v=(p_{\Psi^+}-p_{\Psi^-})/2$. The density-matrix conditions in
\eqref{eq:summary-Xstate} are automatic from $p_\beta\ge0$. The PPT
conditions reduce to
$$
p_{\Psi^+}+p_{\Psi^-}\ge |p_{\Phi^+}-p_{\Phi^-}|,
\qquad
p_{\Phi^+}+p_{\Phi^-}\ge |p_{\Psi^+}-p_{\Psi^-}|,
$$
which is equivalent to $p_\beta\le\tfrac12$ for all four Bell weights.
This recovers the standard Bell-diagonal separability condition.
Substitution in
\eqref{eq:EAC-X-PhiP} gives the linear law
$\EAC(\sigma_{\mathrm{Bell}};\Phi^+)=1-2p_{\Phi^+}$, hence
$$
\lambda_*(\sigma_{\mathrm{Bell}};\Phi^+)
=\frac{1-2p_{\Phi^+}}{2(1-p_{\Phi^+})}.
$$

\subsection{Fidelity absorption capacity}
\label{subsec:xstate-fidelity}
We next evaluate the fidelity absorption capacity $\EACF$, the
teleportation analogue of the preceding $\EAC$ calculation. The
Bell-mixing line is the same, but the stopping condition changes: $\EAC$
is determined by PPT positivity, while $\EACF$ is determined by the first
point at which the fully entangled fraction $f(\rho_\eta)$ reaches
$\tfrac12$. In the $|\PhiP\rangle$ frame, write
$$
\rho_\eta=\frac{\sigma_X+\eta\Phi^+}{1+\eta}.
$$
The only new calculation is $f(\rho_\eta)$. We first record the
formula for a generic $X$ state and then apply it to $\rho_\eta$. Write a
two-qubit vector as $|\beta\rangle=\sum_{i,j}M_{ij}|ij\rangle$. Maximal
entanglement is equivalent to the one-qubit marginal $MM^\dagger$ being
$I_2/2$. This means the two rows of $M$ are orthogonal and each has norm
$1/\sqrt2$. After removing an overall phase, this gives the parametrization
$$
|\beta\rangle
=\frac{1}{\sqrt2}\Bigl(
e^{i\alpha}\sqrt p\,|00\rangle
-e^{-i\gamma}\sqrt{1-p}\,|01\rangle
+e^{i\gamma}\sqrt{1-p}\,|10\rangle
+e^{-i\alpha}\sqrt p\,|11\rangle
\Bigr),
$$
where $0\le p\le 1$ and $\alpha,\gamma\in[0,2\pi)$.
Substituting this vector into a generic $X$ state $\omega_X$ with entries
$(A,B,C,D,U,V)$ gives
$$
\langle\beta|\omega_X|\beta\rangle
=
p\left(\frac{A+D}{2}
+\operatorname{Re}(U e^{-2i\alpha})\right)
+(1-p)\left(\frac{B+C}{2}
-\operatorname{Re}(V e^{2i\gamma})\right).
$$
The maximization over $|\beta\rangle$ is elementary. For fixed $p$,
the two phase-dependent terms can attain the values $|U|$ and $|V|$.
The remaining expression is a convex combination of the two optimized
block overlaps, so maximizing over $0\le p\le1$ selects the larger one:
$$
f(\omega_X)
=
\max\left\{
\frac{A+D}{2}+|U|,\,
\frac{B+C}{2}+|V|
\right\}.
$$
We now set $\omega_X=\rho_\eta$. Its entries are those of
$\sigma_X+\eta\Phi^+$ divided by $1+\eta$, namely
$a+\eta/2,b,c,d+\eta/2,u+\eta/2,v$ divided by $1+\eta$:
$$
(A,B,C,D,U,V)
=\frac{1}{1+\eta}
\left(a+\frac{\eta}{2},\,b,\,c,\,d+\frac{\eta}{2},\,u+\frac{\eta}{2},\,v\right).
$$
Hence
$f(\rho_\eta)=\max\{f_\Phi(\rho_\eta),f_\Psi(\rho_\eta)\}$, where
$$
f_\Phi(\rho_\eta)
=\frac{(a+d)/2+\eta/2+|u+\eta/2|}{1+\eta},
\qquad
f_\Psi(\rho_\eta)
=\frac{(b+c)/2+|v|}{1+\eta}.
$$
The $\Psi$ branch cannot give the first crossing: separability gives
$|v|\le\sqrt{ad}\le(a+d)/2$, hence $f_\Psi(\rho_0)\le\tfrac12$, and the
branch decreases with $\eta$. The crossing therefore comes from
$f_\Phi(\rho_\eta)=\tfrac12$. Using $a+b+c+d=1$, this condition is
equivalent to
$$
|u+\eta/2|=\frac{b+c}{2}.
$$
Solving this equation for $\eta$ gives the largest nonnegative solution
$$
\eta=-2\operatorname{Re}u+\sqrt{(b+c)^2-4(\operatorname{Im}u)^2}.
$$
The four Bell-frame fidelity absorption capacities are therefore
\begin{align}
\EACF(\sigma_X;\Phi^+)
&=-2\operatorname{Re} u+\sqrt{(b+c)^2-4(\operatorname{Im} u)^2},
\label{eq:EAC-F-X}
\\
\EACF(\sigma_X;\Phi^-)
&=2\operatorname{Re} u+\sqrt{(b+c)^2-4(\operatorname{Im} u)^2},
\label{eq:EACF-X-PhiM}\\
\EACF(\sigma_X;\Psi^+)
&=-2\operatorname{Re} v+\sqrt{(a+d)^2-4(\operatorname{Im} v)^2},
\label{eq:EACF-X-PsiP}\\
\EACF(\sigma_X;\Psi^-)
&=2\operatorname{Re} v+\sqrt{(a+d)^2-4(\operatorname{Im} v)^2}.
\label{eq:EACF-X-PsiM}
\end{align}
As capacities on separable $X$ states, these quantities are nonnegative.
The teleportation threshold is
\begin{equation}
\lambda_F(\sigma_X;\Phi^+)
=\frac{\EACF(\sigma_X;\Phi^+)}{1+\EACF(\sigma_X;\Phi^+)}\,.
\label{eq:lambdaF-X}
\end{equation}
Equations \eqref{eq:EAC-F-X}--\eqref{eq:lambdaF-X} are the teleportation
counterpart of the separability formulas
\eqref{eq:EAC-X-PhiP}--\eqref{eq:EAC-X-PsiM}.

\begin{corollary}[Threshold order and gap]
\label{cor:sandwich-gap}
In the $|\PhiP\rangle$ frame, the capacity gap has the exact form
\begin{equation}
\begin{split}
\EACF(\sigma_X;\PhiP)-\EAC(\sigma_X;\PhiP)
={}&
\frac{(b-c)^2}
{\sqrt{(b+c)^2-4(\operatorname{Im}u)^2}
+2\sqrt{bc-(\operatorname{Im}u)^2}},
\end{split}
\label{eq:sandwich-gap-capacity}
\end{equation}
The denominator is nonnegative on separable $X$ states, so
\eqref{eq:sandwich-gap-capacity} gives $\EACF\ge\EAC$, with equality
exactly when $b=c$. Since $x\mapsto x/(1+x)$ is increasing,
\begin{equation}
\lambda_*(\sigma_X;\Phi^+)\,\le\,\lambda_F(\sigma_X;\Phi^+),
\qquad
\text{equality $\,\Longleftrightarrow\, b=c$.}
\label{eq:X-sandwich-equality}
\end{equation}
The exact threshold gap is
\begin{equation}
\lambda_F(\sigma_X;\PhiP)-\lambda_*(\sigma_X;\PhiP)
=\frac{\EACF-\EAC}{(1+\EAC)(1+\EACF)}.
\label{eq:sandwich-gap-threshold}
\end{equation}
For real $u$, the capacity gap reduces to $(\sqrt{b}-\sqrt{c})^2$.
For
$\lambda\in(\lambda_*,\lambda_F)$, the state $\rho_\lambda$ is already
entangled but still below the teleportation threshold. This is the
Bell-weight interval in which the link contains entanglement, but an
additional distillation or filtering step is needed before it is directly
useful for teleportation.
\end{corollary}

\paragraph*{Aligned pure-state references.}
The $1/C_\psi$ rescaling of $\EACF$ applies when $\sigma_X$ has the
$X$ shape in \eqref{eq:summary-Xstate} in the Schmidt basis of the pure reference
$|\psi\rangle=\alpha|00\rangle+\beta|11\rangle$ in
\eqref{eq:psi-schmidt}, with $C_\psi=2\alpha\beta>0$.
Along the mixing line
$\rho_\eta=(\sigma_X+\eta\psi)/(1+\eta)$ the state remains $X$-shaped.
Its entries are those of $\sigma_X+\eta\psi$ divided by $1+\eta$,
namely
$a+\eta\alpha^2,b,c,d+\eta\beta^2,u+\eta\alpha\beta,v$
divided by $1+\eta$:
$$
(A,B,C,D,U,V)
=\frac{1}{1+\eta}
\left(a+\eta\alpha^2,\,b,\,c,\,d+\eta\beta^2,\,
u+\eta\alpha\beta,\,v\right).
$$
Since $\alpha^2+\beta^2=1$,
$f(\rho_\eta)=\max\{f_\Phi(\rho_\eta),f_\Psi(\rho_\eta)\}$, where
$$
f_\Phi(\rho_\eta)
=\frac{(a+d)/2+\eta/2+|u+\eta\alpha\beta|}{1+\eta},
\qquad
f_\Psi(\rho_\eta)
=\frac{(b+c)/2+|v|}{1+\eta}.
$$
Compared with the Bell-reference case, the only change is the replacement
$\eta/2\mapsto \eta\alpha\beta=C_\psi\eta/2$ inside the $\Phi$-block
coherence.
The $\Psi$ branch cannot give the first crossing for the same separability
reason as above. Setting $f_\Phi(\rho_\eta)=\tfrac12$ gives
$$
\left|u+\frac{C_\psi}{2}\eta\right|=\frac{b+c}{2}.
$$
Solving gives the Bell-reference capacity divided by $C_\psi$:
\begin{equation}
\EACF(\sigma_X;\psi)=\frac{\EACF(\sigma_X;\Phi^+)}{C_\psi},
\label{eq:EACF-X-pure-aligned}
\end{equation}
The Bell-weight thresholds then follow from the usual map
$\lambda=\eta/(1+\eta)$.

\paragraph*{Remark: entangled noise $\sigma_X$.}
We consider absorption capacities only for separable noise states
$\sigma_X$. The same algebra above also applies to the line
$(\sigma_X+\eta\Phi^+)/(1+\eta)$ when $\sigma_X$ is already entangled, but
then the result should be read only as a trajectory calculation, not as a
capacity statement. The PPT status is obtained by checking the two block
inequalities
$$
(a+\eta/2)(d+\eta/2)\ge |v|^2,\qquad
bc\ge |u+\eta/2|^2 .
$$
Failure of either inequality means entanglement. Since
$\rho_\eta\to\Phi^+$ as $\eta\to\infty$, when $\eta$ increases from $0$ the
line either remains entangled for all $\eta$ or passes through a single
PPT/separable interval, possibly degenerate to a single point. Convexity of the
separable set implies that there is just one such separable window on the
full Bell-mixing line. Similarly,
teleportation usefulness is decided by
$$
\max\{f_\Phi(\rho_\eta),f_\Psi(\rho_\eta)\}>\frac12 .
$$
The teleportation-useless part of the line is likewise a single interval:
$f(\rho_\lambda)$ is convex in the Bell weight $\lambda=\eta/(1+\eta)$, as
shown in the proof of Theorem~\ref{thm:mobius}, so its sublevel set
$\{\lambda:f(\rho_\lambda)\le 1/2\}$ is an interval.

\section{Decoherence trajectories}
\label{sec:trajectories}
This section uses the closed forms to follow how the thresholds move as
the separable noise state evolves. Let $\sigma(t)=\mathcal N_t(\sigma)$ be
the noise after a local channel $\mathcal N_t$, with the reference in the
mixing line held fixed at $\Phi^+$. We compute $\EAC(\sigma(t))$ and
$\EACF(\sigma(t))$, the amounts of the fixed $\Phi^+$ reference the evolved
noise absorbs before crossing the entanglement and teleportation
boundaries, with the corresponding Bell-weight thresholds $\lambda_*(t)$
and $\lambda_F(t)$. Propagating the whole Bell-mixed state through
$\mathcal N_t$ is a different question, since the reference component would
become $\mathcal N_t(\Phi^+)$ rather than the fixed reference $\Phi^+$.

If the local channel keeps the noise in the $X$ family, the two thresholds
are obtained by substituting the evolved entries into
\eqref{eq:EAC-X-PhiP} and \eqref{eq:EAC-F-X}. We illustrate this with two
local channels that preserve $X$ form. Amplitude damping is the
channel $\mathcal A_\gamma(\rho)=K_0\rho K_0^\dagger+K_1\rho K_1^\dagger$
with Kraus operators
$K_0=|0\rangle\!\langle0|+\sqrt{1-\gamma}\,|1\rangle\!\langle1|$ and
$K_1=\sqrt{\gamma}\,|0\rangle\!\langle1|$. It transfers population from
$|1\rangle$ to $|0\rangle$ and damps coherences. Pure dephasing is
$\mathcal Z_p(\rho)=(1-p)\rho+p\,\sigma_z\rho\sigma_z$. It leaves
populations fixed and damps only the off-diagonal entries. The time
constants $T_1$ and $T_\varphi$ enter through
$\gamma(t)=1-e^{-t/T_1}$ and $1-2p(t)=e^{-t/T_\varphi}$.
When both processes are present, the total coherence time $T_2$ is defined by
$1/T_2=1/(2T_1)+1/T_\varphi$.
Figure~\ref{fig:xstate-AD} first shows the two effects separately as
functions of their channel strengths. Figure~\ref{fig:realistic-noise}
then uses the combined time-dependent entries.

\begin{figure*}[!t]
\centering
\includegraphics[width=0.92\textwidth]{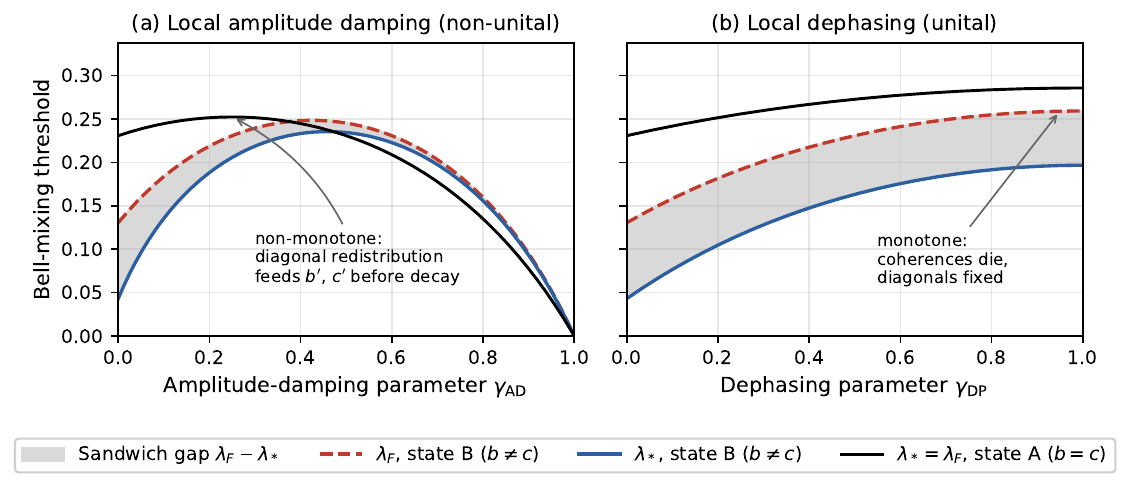}
\caption{Bell-mixing thresholds for separable $X$ noise evolving under local
amplitude damping and pure dephasing. Each panel shows the entanglement
threshold $\lambda_*$ and the teleportation threshold $\lambda_F$ as functions
of the channel strength. Noise state~A starts with equal middle populations $b=c$, and the
channels preserve this along the whole evolution, so the two thresholds
coincide by \eqref{eq:X-sandwich-equality}. For noise state~B the middle populations are
unequal, giving the shaded interval where the Bell mixture is entangled but
still below the teleportation threshold. Amplitude damping changes the
diagonal populations and can make the thresholds rise and fall, whereas pure dephasing
leaves the populations fixed and damps only the coherences. The curves are
obtained by substituting the evolved $X$ entries into \eqref{eq:EAC-X-PhiP}
and \eqref{eq:EAC-F-X}.}
\label{fig:xstate-AD}
\end{figure*}

\paragraph*{Amplitude-damped and dephased \texorpdfstring{$X$-state}{X-state} registers.}
The two channels used in Fig.~\ref{fig:xstate-AD} act locally on both
qubits and preserve the $X$ block structure. For symmetric amplitude
damping of strength $\gamma$, the map
$\mathcal A_\gamma\otimes\mathcal A_\gamma$ sends
\begin{equation*}
\begin{pmatrix}
a&0&0&u\\ 0&b&v&0\\ 0&v^{*}&c&0\\ u^{*}&0&0&d
\end{pmatrix}
\,\xrightarrow{\,\mathcal A_\gamma\otimes\mathcal A_\gamma\,}\,
\begin{pmatrix}
a'&0&0&u'\\ 0&b'&v'&0\\ 0&v'^{*}&c'&0\\ u'^{*}&0&0&d'
\end{pmatrix},
\end{equation*}
with
\begin{align}
a'&=a+\gamma(b+c)+\gamma^2 d,
\nonumber\\
b'&=(1-\gamma)(b+\gamma d),\qquad
c'=(1-\gamma)(c+\gamma d),
\label{eq:AD-X-map}
\\
d'&=(1-\gamma)^2 d,\qquad
u'=(1-\gamma)u,\qquad
v'=(1-\gamma)v.
\nonumber
\end{align}
Substituting these entries into \eqref{eq:EAC-X-PhiP} and
\eqref{eq:EAC-F-X}--\eqref{eq:lambdaF-X} gives
$\lambda_*(\gamma)$ and $\lambda_F(\gamma)$ directly.

For symmetric local dephasing $\mathcal Z_p\otimes\mathcal Z_p$, the
populations are unchanged and only the coherences contract:
$$
\sigma_X\ \xrightarrow{\,\mathcal Z_p\otimes\mathcal Z_p\,}\
\begin{pmatrix}
a&0&0&(1-2p)^2u\\
0&b&(1-2p)^2v&0\\
0&(1-2p)^2v^{*}&c&0\\
(1-2p)^2u^{*}&0&0&d
\end{pmatrix},
$$
so the same closed forms give $\lambda_*(p)$ and $\lambda_F(p)$ after this
replacement.

When amplitude damping and dephasing act together, as in
Fig.~\ref{fig:realistic-noise}, the populations are those of amplitude
damping alone, since dephasing leaves populations fixed, and the coherences
carry both the amplitude-damping factor $(1-\gamma)$ and the dephasing
factor $(1-2p)^2$:
\begin{align}
a_t&=a+\gamma(b+c)+\gamma^2 d,
\nonumber\\
b_t&=(1-\gamma)(b+\gamma d),\qquad
c_t=(1-\gamma)(c+\gamma d),
\label{eq:AD-dephasing-X-map}
\\
d_t&=(1-\gamma)^2 d,\qquad
u_t=(1-\gamma)(1-2p)^2u,\qquad
v_t=(1-\gamma)(1-2p)^2v.
\nonumber
\end{align}
With $\gamma=\gamma(t)$ and $p=p(t)$, these entries make
$\lambda_*(t)$ and $\lambda_F(t)$ explicit functions of time. Since
amplitude damping can transiently increase $b_tc_t$, neither $\EAC$ nor
$\EACF$ is forced to be monotone in the damping strength.
Figure~\ref{fig:xstate-AD} shows this behavior on representative
separable $X$-state noise states.

\paragraph*{Closed-form threshold diagram.}
The $X$-state formulas yield a threshold diagram along a noise trajectory.
Figure~\ref{fig:realistic-noise} shows a separable $X$ noise state under
combined amplitude damping and dephasing, evaluated from the time-dependent
entries \eqref{eq:AD-dephasing-X-map}. After substituting
$\gamma(t)$ and $p(t)$, the formulas give $\lambda_*(t)$ and
$\lambda_F(t)$ directly as functions of time. A fixed Bell weight then
crosses the entanglement or teleportation threshold at the solutions of
ordinary equations in $t$. The channels are local, so a separable initial
noise state remains separable and the separable-state formulas apply along
the whole trajectory.

\begin{figure*}[!t]
\centering
\includegraphics[width=0.7\textwidth]{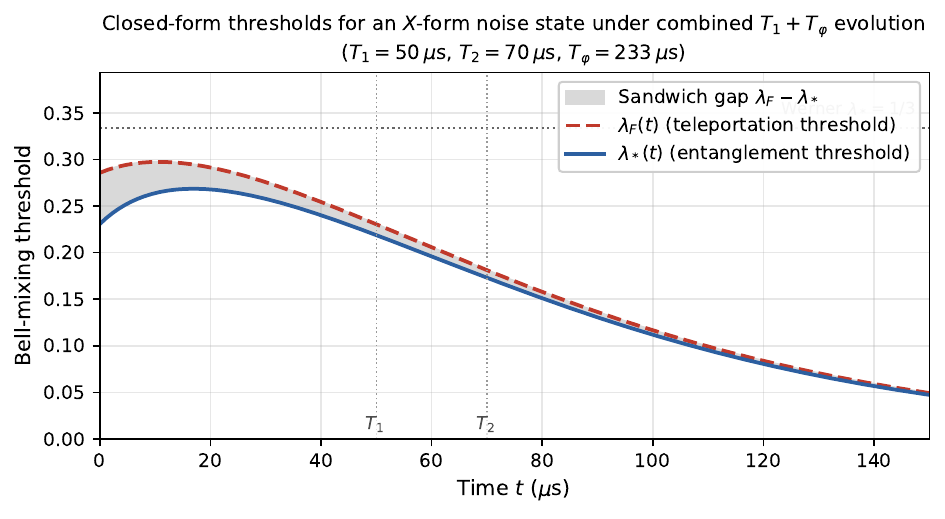}
\caption{Closed-form entanglement and teleportation thresholds for a
separable $X$ noise state under combined local amplitude damping and pure
dephasing. The illustrative parameters are $T_1=50\,\mu$s and
$T_2=70\,\mu$s, giving
$T_\varphi=1/(1/T_2-1/(2T_1))\approx 233\,\mu$s, in the transmon range
\cite{Krantz2019}. The noise state has biased middle populations
($b\neq c$), so the threshold gap is nonzero. Both curves are evaluated
from the combined entries in \eqref{eq:AD-dephasing-X-map} using
\eqref{eq:EAC-X-PhiP} and \eqref{eq:EAC-F-X}. The gap closes as the
middle-population asymmetry is
removed by relaxation, and the early-time nonmonotonicity comes from
amplitude damping transiently increasing $b_tc_t$. Werner's
$\lambda_*=1/3$ for $\sigma=I_4/4$ is shown for reference.}
\label{fig:realistic-noise}
\end{figure*}

\section{Discussion and outlook}
\label{sec:outlook}
The construction in this paper keeps the reference state fixed and treats
the separable noise state as the input. The capacities $\EAC$ and $\EACF$
then give the two Bell-mixing thresholds through the same M\"obius map.
On product noise and separable $X$ noise, these quantities become closed
forms, so the entanglement threshold, the teleportation threshold, and the
entangled-but-not-yet-useful interval can be compared directly.

\subsection{Fixed-reference interpretation}
\label{subsec:relation-existing-measures}
The absorption capacities are fixed-reference quantities. They do not
measure the entanglement of a state by optimizing over all directions.
Instead, they fix the noise state and the reference state, and ask where
that particular mixing line crosses the entanglement or teleportation
threshold. In this sense they are closer to relative robustness than to
state measures such as concurrence or negativity.

The closest standard object is the Vidal--Tarrach robustness of
entanglement relative to a fixed separable state~\cite{VidalTarrach1999}.
If $R(\rho\,\|\,\sigma)$ is the least $t\ge0$ for which
$(\rho+t\sigma)/(1+t)$ is separable, then for the $|\Phi^+\rangle$
reference
\begin{equation}
\EAC(\sigma)=\frac{1}{R(\PhiP\,\|\,\sigma)}.
\label{eq:EAC-relative-robustness}
\end{equation}
The two quantities describe the same mixing line, written in reciprocal
variables. Other robustness variants answer different questions, since
they either fix a different noise state or optimize over allowed noise
states~\cite{Steiner2003}. Here the noise state and the reference are both
part of the question.

\subsection{Decomposition certificates}
\label{subsec:decomposition-certificates}
The reciprocal relation to robustness is conceptually useful, but the
capacity variable is more convenient for decomposition certificates.
Feasible absorption amounts add linearly: each component can absorb its
own reference term, and the convex sum absorbs the weighted sum. Suppose
$\sigma=\sum_i p_i\sigma_i$. If $\eta_i$ is feasible for $\sigma_i$ for
each $i$, then $\sum_i p_i\eta_i$ is feasible for $\sigma$:
\begin{equation}
\sigma^{T_B}+\Bigl(\sum_i p_i\eta_i\Bigr)Q
=
\sum_i p_i\bigl(\sigma_i^{T_B}+\eta_i Q\bigr)\succeq0.
\label{eq:EAC-concavity-certificate}
\end{equation}
Consequently
\begin{equation}
\EAC(\sigma)\ge \sum_i p_i\EAC(\sigma_i),
\label{eq:EAC-concavity}
\end{equation}
with the usual limiting interpretation when a supremum is not attained.
We call this a certificate because the displayed decomposition itself is a
proof that the weighted absorption value is feasible for $\sigma$, even if
it is not the exact capacity.
The fidelity absorption capacity obeys the same certificate because its
condition $f(\rho_\eta)\le\tfrac12$ is an intersection of affine
constraints in $(\sigma,\eta)$ over maximally entangled test states. Hence
\begin{equation}
\EACF(\sigma)\ge \sum_i p_i\,\EACF(\sigma_i).
\label{eq:EACF-concavity}
\end{equation}
A product or $X$ decomposition gives a certified lower bound on both
thresholds. The certificate need not be tight, since the mixture can
absorb more reference than the weighted componentwise amount.

\subsection{Concurrence and negative entanglement measure}
\label{subsec:optimized-margins}
On real $X$ states, the connection with concurrence and the negative
entanglement measure of Zhang \emph{et al.}~\cite{ZhangHanZhang2010} is
visible at the level of the two Bell blocks.
Define the four signed quantities
$$
\Delta_{\Phi^\pm}=2(\pm u-\sqrt{bc}),\qquad
\Delta_{\Psi^\pm}=2(\pm v-\sqrt{ad}).
$$
The usual concurrence keeps the largest positive one,
$$
C(\rho_X)
=\max\{0,\Delta_{\Phi^+},\Delta_{\Phi^-},\Delta_{\Psi^+},\Delta_{\Psi^-}\}
=\max\{0,2(|u|-\sqrt{bc}),2(|v|-\sqrt{ad})\}.
$$
For separable $\sigma_X$ all four signed quantities are nonpositive. The
matched-frame capacities are the corresponding positive values
$-\Delta_\beta$,
$$
\EAC(\sigma_X;\Phi^\pm)=-\Delta_{\Phi^\pm},\qquad
\EAC(\sigma_X;\Psi^\pm)=-\Delta_{\Psi^\pm}.
$$
The negative entanglement measure minimizes these positive values over
Bell frames:
\begin{equation}
|N(\sigma_X)|
=\min_{\beta\in\{\Phi^\pm,\Psi^\pm\}}\EAC(\sigma_X;\beta)
=\min\bigl\{\,2(\sqrt{bc}-|u|),\ 2(\sqrt{ad}-|v|)\,\bigr\}.
\label{eq:nem-min}
\end{equation}
Concurrence uses the largest positive value among the four
$\Delta_\beta$. On the separable side, the negative entanglement measure
uses the smallest positive value among the four $-\Delta_\beta$. By
contrast, the entanglement absorption
capacity keeps the Bell frame fixed because the Bell-mixing threshold
fixes the resource state. For complex $u,v$, it also retains the phase
dependence visible in \eqref{eq:EAC-X-PhiP}--\eqref{eq:EAC-X-PsiM}.

\subsection{Open directions}
\paragraph*{Beyond separable noise states.}
In this paper the noise state is separable, so $\eta=0$ is feasible in
\eqref{eq:EAC-def} and the capacity measures how much reference can be
added before separability is lost. If the noise state is already
entangled, the same line no longer starts from the separable set, so
$\EAC$ is not an entanglement-threshold quantity without a new
convention. Whether a useful fixed-reference analogue can be defined on
the entangled side remains open.

\paragraph*{Beyond the solved separable families.}
The threshold identity \eqref{eq:mobius} and the variational evaluation
\eqref{eq:EAC-Rayleigh} are universal on the separable set, but the
closed forms in Sections~\ref{sec:eval1}--\ref{sec:trajectories} rely on
product or $X$-state structure. A natural next step is to identify other
separable families where the positivity or teleportation condition
reduces to a comparable closed form. For a generic separable state,
decomposition certificates from \eqref{eq:EAC-concavity} and
\eqref{eq:EACF-concavity} give analytic lower bounds.

\paragraph*{Optimizing decomposition certificates.}
Given a solved class $\mathcal C$, one can optimize the certificate over
all decompositions of $\sigma$ into components from that class:
$$
L_{\mathcal C}(\sigma)
:=
\sup_{\sigma=\sum_i p_i\sigma_i,\ \sigma_i\in\mathcal C}
\sum_i p_i\,\EAC(\sigma_i),
$$
and analogously for $\EACF$. The equality problem is to identify
conditions under which $\sigma$ has an optimal decomposition whose
components saturate the same capacity constraint, so that
$L_{\mathcal C}(\sigma)=\EAC(\sigma)$. Another concrete question is how
tight the certificate becomes when $\mathcal C$ is the product family,
the separable $X$ family, or their union.

\paragraph*{Higher dimensions and multipartite systems.}
The fixed-reference mixing question can also be asked in $d\times d$,
with $|\Phi_d\rangle=\sum_i|ii\rangle/\sqrt d$, but the capacity must
then be defined directly by separability:
$$
\EAC^{(d)}(\sigma;\Phi_d)
:=\sup\left\{\eta\ge0:
\frac{\sigma+\eta\Phi_d}{1+\eta}\in\Sep_{d,d}\right\}.
$$
On the isotropic line with maximally mixed noise, symmetry gives
$\EAC^{(d)}(I_{d^2}/d^2;\Phi_d)=1/d$, equivalently
$\lambda_*=1/(d+1)$. Beyond such special symmetric families, however,
the two-qubit method does not carry over directly. For $d\ge3$,
positivity of the partial transpose is not equivalent to separability,
and the product impurity law also fails: $\EAC^{(d)}(A\otimes B;\Phi_d)$ is
not determined by the local purities alone. Higher-dimensional and
multipartite versions, such as GHZ-mixing thresholds, therefore require
new ideas beyond the special symmetric cases above.

\subsection{Data availability}
This work contains no measured data. Every figure is a deterministic plot of
the closed-form expressions in the text, evaluated at fixed illustrative
parameters. Figure~\ref{fig:product-threshold} plots the product-noise closed
form over the full range of local Bloch radii and needs no further input.
Figure~\ref{fig:xstate-AD} uses two separable $X$ noise states, each specified
by its populations $(a,b,c,d)$ and the two coherences $u$ and $v$ of the $X$
matrix: state~A has $(a,b,c,d)=(0.30,0.20,0.20,0.30)$ with $u=v=0.05$ (equal
middle populations $b=c$), and state~B has $(a,b,c,d)=(0.15,0.05,0.30,0.50)$
with $u=0.10$, $v=0.05$. Figure~\ref{fig:realistic-noise} evolves the
separable $X$ noise state $(a,b,c,d)=(0.20,0.10,0.40,0.30)$ with $u=v=0.05$
under $T_1=50\,\mu$s and $T_2=70\,\mu$s. These values are fixed choices, not
sampled at random, and reproduce the plotted curves exactly through the
closed forms cited in each caption.

\section*{Author contributions}
Xuan Du Trinh is the sole author, responsible for all aspects of the work,
including its conception, the analytical results with their proofs and figures,
and the writing of the manuscript. Anthropic's Claude was
used as an assistant to clean up the LaTeX syntax, especially for tables, and
to produce figures. The author verified all
content and takes full responsibility for it.

\section*{Acknowledgements}
The author thanks Khanh Gia Nguyen for enjoyable discussions.

\bibliographystyle{unsrtnat}
\bibliography{refs_Quantum}

\appendix

\section*{Overview of the appendices}

The appendices contain the detailed technical proofs that support the results
stated in the main text.
Appendix~\ref{sec:supp-proof} proves the product-state law and its
pure-reference extension. Appendix~\ref{sec:supp-sandwich-notes} gives
the product-noise teleportation calculation for Bell and pure references,
including the two orientation examples used in the main text.
Appendix~\ref{sec:supp-channel-calculus} gives the product-family
channel calculations and the unital monotonicity statement. Appendix~\ref{subsec:evaluating-capacities}
gives the general evaluation formulas used outside the product and $X$
closed forms. Appendix~\ref{sec:supp-extra-closed-forms}
collects additional closed forms.

\section{Proof of the product-state law}
\label{sec:supp-proof}

This appendix proves Lemma~\ref{lem:bell-block}, the product-state inequality
used in the main text. For any single-qubit density matrices $A,B$,
\begin{equation}\label{eq:supp-bellblock}
M := A\otimes B^T+2\sqrt{\det A\,\det B}\,Q\,\succeq\,0,
\end{equation}
and the coefficient $2\sqrt{\det A\,\det B}$ is sharp for that fixed pair
$(A,B)$. Here $Q=(|\PhiP\rangle\!\langle\PhiP|)^{T_B}=\Swap/2$ and
$\Swap=\sum_{ij}|ij\rangle\!\langle ji|$ is the swap operator on
$\mathbb C^2\otimes\mathbb C^2$. Writing $t:=\sqrt{\det A\,\det B}$,
\eqref{eq:supp-bellblock} is equivalent to $M=A\otimes B^T+t\Swap\succeq 0$.
As a standalone matrix inequality, the transpose is inessential. Replacing
$B^T$ by $B$ gives the same statement, since transposition preserves
positivity and determinant. We keep $B^T$ to match the PPT calculation, where
$(A\otimes B)^{T_B}=A\otimes B^T$.
The proof uses the $U\otimes U$ symmetry of $\Swap$ to reduce to diagonal
$A$, then checks positivity by a Schur complement. The same calculation shows
sharpness. Reading off the largest admissible coefficient gives the product
law
$$
\EAC(A\otimes B)=2\sqrt{\det A\det B}
=\sqrt{(1-\Tr A^2)(1-\Tr B^2)}.
$$

\subsection{Proof of the Bell-block inequality}

\paragraph*{Rank-deficient case.}
If $\det A=0$ or $\det B=0$, the claimed coefficient is zero and
$A\otimes B^T\succeq0$. It remains to show that no positive coefficient
can be added. Write $M(t)=A\otimes B^T+t\Swap$. Suppose first that
$\det A=0$ and diagonalize $A=\operatorname{diag}(a,0)$ with $a>0$. Write
$B^T=(b_{ij})_{i,j=0}^1$. In the computational basis $00,01,10,11$,
$$
M(t)=
\begin{pmatrix}
a b_{00}+t & a b_{01} & 0 & 0\\
a b_{10} & a b_{11} & t & 0\\
0 & t & 0 & 0\\
0 & 0 & 0 & t
\end{pmatrix},
$$
so the principal submatrix on $\{|01\rangle,|10\rangle\}$ has determinant
$-t^2<0$ for every $t>0$. Hence $M(t)$ is not positive semidefinite for
any $t>0$. The case $\det B=0$ is analogous, with the roles of the two
tensor factors interchanged. Therefore the largest admissible coefficient
is zero, as claimed.

\paragraph*{Reduction to diagonal $A$.}
We exploit the swap operator's $U\otimes U$ invariance to diagonalize
$A$ for free, without disturbing the $Q$ term. Specifically,
$\Swap$ commutes with $U\otimes U$ for every unitary $U$. Choosing $U$ so that
$UAU^\dagger=\operatorname{diag}(a,d)$ with $ad=\det A$, the conjugation
$$
(U\otimes U)\,M\,(U\otimes U)^\dagger
=\operatorname{diag}(a,d)\otimes\bigl(UB^TU^\dagger\bigr)+t\Swap
$$
preserves both positivity and the swap term.  The new ``$B^T$'' factor
$B':=UB^TU^\dagger$ is again a positive Hermitian matrix with
$\det B'=\det B$ (similarity preserves the determinant). We parameterize it as
\begin{equation*}
B'=\begin{pmatrix} p & \bar q\\ q & r\end{pmatrix},
\end{equation*}
with $p,r>0$, $p+r=1$, and $pr-|q|^2=\det B$.
It therefore suffices to prove $M\succeq 0$ for
$A=\operatorname{diag}(a,d)$ and $B^T$ replaced by $B'$ of the above form.

\paragraph*{Matrix form.}
In the computational basis ordered as $00,01,10,11$,
the nonzero entries of $\Swap$ are
$\Swap_{(00)(00)}=\Swap_{(11)(11)}=1$ and
$\Swap_{(01)(10)}=\Swap_{(10)(01)}=1$. Since
$A\otimes B'=\operatorname{diag}(a,d)\otimes B'$, we get
\begin{equation}\label{eq:supp-M}
M=\begin{pmatrix}
ap+t & a\bar q & 0 & 0\\
aq   & ar     & t & 0\\
0    & t      & dp & d\bar q\\
0    & 0      & dq & dr+t
\end{pmatrix}.
\end{equation}
The four added $t$ terms are exactly the nonzero entries of $t\Swap$.

\paragraph*{Block-permuted decomposition.}
Reorder the basis to $\{|00\rangle,|11\rangle,|01\rangle,|10\rangle\}$ by the
permutation $\Pi$ mapping $(v_{00},v_{01},v_{10},v_{11})^T$ to
$(v_{00},v_{11},v_{01},v_{10})^T$. In this order, all nonzero entries of
$t\Swap$ lie inside the diagonal blocks. Explicitly,
$$
M'=\Pi M \Pi^T=
\left(\begin{array}{cc|cc}
ap+t & 0 & a\bar q & 0\\
0 & dr+t & 0 & dq\\
\hline
aq & 0 & ar & t\\
0 & d\bar q & t & dp
\end{array}\right)
=\begin{pmatrix} P & R\\ R^\dagger & W\end{pmatrix},
$$
where the vertical and horizontal lines only partition the rows and columns
after the first two basis vectors. With this partition, $M'$ is a $2\times2$ block matrix
with
$P=\operatorname{diag}(ap+t,\,dr+t)$,
$R=\operatorname{diag}(a\bar q,\,dq)$, and
$W=\bigl(\begin{smallmatrix} ar & t\\ t & dp\end{smallmatrix}\bigr)$. The
diagonal entries of $P$ are positive, so the Schur complement criterion
applies. Since $P$ is diagonal,
$P^{-1}=\operatorname{diag}\!\bigl(\tfrac1{ap+t},\,\tfrac1{dr+t}\bigr)$.

\paragraph*{Schur complement.}
A direct computation gives
$$
R^\dagger P^{-1}R
=\operatorname{diag}\!\Bigl(\tfrac{a^2|q|^2}{ap+t},\,\tfrac{d^2|q|^2}{dr+t}\Bigr),
$$
so the Schur complement is
$$
W-R^\dagger P^{-1}R
=\begin{pmatrix}
\dfrac{a^2\det B+art}{ap+t} & t\\[6pt]
t & \dfrac{d^2\det B+dpt}{dr+t}
\end{pmatrix},
$$
where we used $ar(ap+t)-a^2|q|^2=a^2(pr-|q|^2)+art=a^2\det B+art$, and
analogously for the second diagonal entry.

\paragraph*{Determinant at the critical $t$.}
Let $\Delta=\det B$ and
$$
t_*=\sqrt{ad\Delta}=\sqrt{\det A\,\det B}.
$$
The determinant of the Schur complement is
$$
\frac{ad(a\Delta+rt)(d\Delta+pt)}{(ap+t)(dr+t)}-t^2.
$$
Since $(ap+t)(dr+t)>0$, it is enough to evaluate the numerator
$$
N(t)=ad(a\Delta+rt)(d\Delta+pt)-t^2(ap+t)(dr+t).
$$
Using $t_*^2=ad\Delta$, direct expansion gives the factorization
\begin{equation}\label{eq:supp-Nfactor}
N(t)=\bigl(t_*^2-t^2\bigr)\bigl[t_*^2+t^2+(ap+dr)t\bigr].
\end{equation}
In particular $N(t_*)=0$, so the determinant of the Schur complement
vanishes at $t=t_*$.

At $t=t_*$, the Schur complement is therefore a $2\times 2$ Hermitian matrix
with positive trace and zero determinant, so it is positive semidefinite of
rank one. Combined with $P\succ 0$, the Schur-complement criterion yields
$M(t_*)\succeq 0$.

\paragraph*{Sharpness of the inequality.}
It remains to show that no larger $t$ is feasible. In
\eqref{eq:supp-Nfactor}, the second factor is strictly positive for all
$t>0$ since $a,d,p,r>0$. Hence $N(t)<0$ for
every $t>t_*$. The Schur complement then has negative determinant and is
not positive semidefinite, so $M(t)$ is not positive semidefinite. Therefore
the largest admissible coefficient of $\Swap$ is $t_*$. Since $Q=\Swap/2$,
the largest admissible coefficient of $Q$ is $2t_*$, giving
$$
\EAC(A\otimes B)\,=\,2\sqrt{\det A\,\det B}
$$
for the fixed product state $A\otimes B$.  $\square$

\subsection{Filtered Bell-block inequality for arbitrary pure-state references}
\label{subsec:supp-filtered-block}

This subsection proves Lemma~\ref{lem:filtered-bell-block}, the
pure-reference extension of the product-state law.

Let $|\psi\rangle=\alpha|00\rangle+\beta|11\rangle$ with
$\alpha,\beta>0$ and $\alpha^2+\beta^2=1$, so the concurrence is
$C_\psi=2\alpha\beta$, and let $R=\mathrm{diag}(\alpha,\beta)$. Using
$|\psi\rangle=\sqrt2(I\otimes R)|\PhiP\rangle$ and $R^T=R$,
\begin{equation}
Q_\psi
:=(|\psi\rangle\!\langle\psi|)^{T_B}
=2(I\otimes R)\,Q\,(I\otimes R),
\label{eq:supp-Qpsi-via-R}
\end{equation}
where $Q=(|\PhiP\rangle\!\langle\PhiP|)^{T_B}$.

Set $\tilde B := R^{-1}BR^{-1}$. Since $R$ is real diagonal with
positive entries, $\tilde B\succeq 0$ and
$$
\det\tilde B = \frac{\det B}{(\alpha\beta)^2} = \frac{\det B}{(C_\psi/2)^2}.
$$
Applying the product-state law (Lemma~\ref{lem:bell-block} of
Section~\ref{sec:eval1}) to the pair $(A,\tilde B)$,
\begin{equation}
A\otimes\tilde B^T + 2\sqrt{\det A\,\det\tilde B}\,Q\,\succeq\,0.
\label{eq:supp-filtered-step1}
\end{equation}
Congruence by the Hermitian matrix $I\otimes R$ preserves positive
semidefiniteness: for any $X\succeq0$, writing $w=(I\otimes R)v$ gives
$v^\dagger(I\otimes R)X(I\otimes R)v=w^\dagger X w\ge0$. The two terms of
\eqref{eq:supp-filtered-step1} transform as
$$
\begin{aligned}
(I\otimes R)(A\otimes\tilde B^T)(I\otimes R)
&=A\otimes(R\tilde B^T R)=A\otimes B^T,\\[3pt]
2\sqrt{\det A\det\tilde B}\,(I\otimes R)\,Q\,(I\otimes R)
&=\frac{2\sqrt{\det A\det B}}{C_\psi}\,Q_\psi,
\end{aligned}
$$
using $\tilde B^T=R^{-1}B^TR^{-1}$ in the first and
\eqref{eq:supp-Qpsi-via-R} with $\sqrt{\det\tilde B}=2\sqrt{\det B}/C_\psi$
in the second. Conjugating \eqref{eq:supp-filtered-step1} therefore gives
\begin{equation}
A\otimes B^T + \frac{2\sqrt{\det A\det B}}{C_\psi}\,Q_\psi
\,\succeq\,0,
\label{eq:supp-filtered-final}
\end{equation}
the filtered Bell-block inequality \eqref{eq:filtered-bellblock} of the main
text.

\paragraph*{Sharpness of the inequality.}
The coefficient $2\sqrt{\det A\det B}/C_\psi$ is sharp. Suppose that
$A\otimes B^T+cQ_\psi\succeq0$ for some
$c>2\sqrt{\det A\det B}/C_\psi$. Applying the congruence
$I\otimes R^{-1}$ gives
$$
A\otimes\tilde B^T+2c\,Q\succeq0.
$$
However, Lemma~\ref{lem:bell-block} shows that the largest admissible
coefficient of $Q$ for the pair $(A,\tilde B)$ is
$2\sqrt{\det A\,\det\tilde B}$, while
$$
2c>\frac{4\sqrt{\det A\det B}}{C_\psi}
=2\sqrt{\det A\,\det\tilde B},
$$
which is a contradiction. When $C_\psi=0$, the reference is a product
state and the Bell-mixing line has no entanglement crossing.

\section{Teleportation-threshold derivations}
\label{sec:supp-sandwich-notes}

This appendix gives the product-noise teleportation calculation used in
the main text. For a Bell reference, the Horodecki formula gives a
quadratic equation for the Bell-mixing parameter. For a general pure-state
reference, the same calculation gives a quartic equation. The determinant
of the correlation tensor selects the branch of the Horodecki formula. The
threshold comes from the nonpositive-determinant branch, while the
positive-determinant branch gives no crossing. When $C_\psi=1$, the quartic
equation factors and recovers the Bell-reference quadratic equation. The
full separable $X$-state formula and the threshold gap are derived in
Section~\ref{subsec:xstate-fidelity}.

\paragraph*{Teleportation threshold (Horodecki formula).}
For any two-qubit state $\rho$ with fully entangled fraction
$$
f(\rho):=\max_{U_A,U_B}
\langle\PhiP|(U_A\otimes U_B)\rho(U_A\otimes U_B)^\dagger|\PhiP\rangle,
$$
the maximal teleportation fidelity that $\rho$ can provide is~\cite{HorodeckiTeleportationBell1996,HHH96tele}
\begin{equation}
F(\rho)=\frac{2f(\rho)+1}{3}.
\label{eq:supp-horodecki-F}
\end{equation}
Equivalently, if $T_{ij}=\Tr(\rho\,\sigma_i\otimes\sigma_j)$, with $\sigma_1,\sigma_2,\sigma_3$ the Pauli matrices, is the
correlation tensor and $s_1\ge s_2\ge s_3\ge 0$ are the singular
values of $T$, then
\begin{equation}
f(\rho)=
\begin{cases}
\dfrac{1+s_1+s_2+s_3}{4}, & \det T<0,\\[6pt]
\dfrac{1+s_1+s_2-s_3}{4}, & \det T\ge 0.
\end{cases}
\label{eq:supp-horodecki-f-singvals}
\end{equation}
At $\det T=0$ one has $s_3=0$, so the two branches agree.
The classical teleportation bound $F=2/3$ is therefore the level set
$f(\rho)=1/2$. The Bell-mixing teleportation threshold is
\begin{equation}
\lambda_F(\sigma):=\inf\bigl\{\lambda\in[0,1]:\ f(\rho_\lambda)> \tfrac12\bigr\},
\label{eq:supp-lambdaF-def}
\end{equation}
with $\rho_\lambda=\lambda|\PhiP\rangle\!\langle\PhiP|+(1-\lambda)\sigma$ as in
the main text.
The Bell-state correlation tensor is $T_{\PhiP}=\mathrm{diag}(1,-1,1)$, so
$T_{\rho_\lambda}=\lambda T_{\PhiP}+(1-\lambda)T_\sigma$ is affine in
$\lambda$, but $f(\rho_\lambda)$ is in general nonlinear because it
depends on the singular-value data of $T_{\rho_\lambda}$ together with
the sign of $\det T_{\rho_\lambda}$.

\paragraph*{Universal inequality.}
Let $|\beta\rangle$ be any maximally entangled two-qubit vector. There
exist local unitaries $U,V$ such that
$|\beta\rangle=(U\otimes V)|\PhiP\rangle$. Hence the largest product
overlap with $|\beta\rangle$ is the same as the largest product overlap
with $|\PhiP\rangle$. For normalized one-qubit vectors
$|x\rangle=x_0|0\rangle+x_1|1\rangle$ and
$|y\rangle=y_0|0\rangle+y_1|1\rangle$,
$$
|\langle\PhiP|x\otimes y\rangle|^2
=\frac12\,|\bar x_0\bar y_0+\bar x_1\bar y_1|^2
\le \frac12,
$$
by Cauchy-Schwarz. Therefore
$|\langle\beta|a\otimes b\rangle|^2\le1/2$ for every product vector
$|a\rangle\otimes|b\rangle$.
If $\rho$ is separable, write
$\rho=\sum_k p_k |a_k\rangle\!\langle a_k|\otimes
|b_k\rangle\!\langle b_k|$. Then
$$
\langle\beta|\rho|\beta\rangle
=\sum_k p_k |\langle\beta|a_k\otimes b_k\rangle|^2
\le\frac12 .
$$
Taking the maximum over maximally entangled $|\beta\rangle$ gives
$f(\rho)\le1/2$ on every separable state. Therefore any $\rho$ with
$f(\rho)>1/2$, equivalently $F(\rho)>2/3$, is entangled.
For two qubits, entanglement is equivalent to nonpositive partial
transpose. Along the Bell-mixing line, the partial-transpose crossing is
at $\lambda_*$ by Theorem~\ref{thm:mobius}. Hence the teleportation
threshold cannot be smaller than the entanglement threshold: for every
separable noise state $\sigma$,
\begin{equation}
\lambda_*(\sigma)\le\lambda_F(\sigma).
\label{eq:supp-sandwich}
\end{equation}
The next subsections determine when equality holds in the structured
families used in the main text.

\subsection{Product-family teleportation threshold}
\label{subsec:supp-product-tele}

For product noise $\sigma=A\otimes B$ with
$A=(I+\vec a\!\cdot\!\boldsymbol\sigma)/2$,
$B=(I+\vec b\!\cdot\!\boldsymbol\sigma)/2$, the correlation tensor is the
rank-one matrix $T_\sigma=\vec a\,\vec b^{\,T}$. The Bell state $\PhiP$ has
maximally mixed one-qubit marginals, so it has no one-qubit Pauli terms.
Along the mixing line
$\rho_\eta=(\sigma+\eta\,\PhiP)/(1+\eta)$ in the variable
$\eta=\lambda/(1-\lambda)$,
$$
T_{\rho_\eta}=\frac{\eta D_{\PhiP}+\vec a\,\vec b^{\,T}}{1+\eta},
\qquad D_{\PhiP}=\operatorname{diag}(1,-1,1).
$$
Set
$$
M_\eta:=\eta D_{\PhiP}+\vec a\,\vec b^{\,T}.
$$
The formula \eqref{eq:supp-horodecki-f-singvals} needs the singular values
of $T_{\rho_\eta}=M_\eta/(1+\eta)$, that is, those of $M_\eta$ scaled by
$1/(1+\eta)$. Left-multiplying by the orthogonal matrix $D_{\PhiP}$
preserves singular values, so with $\vec c:=D_{\PhiP}\vec a$ we analyze
$$
N_\eta:=D_{\PhiP}M_\eta=\eta I+\vec c\,\vec b^{\,T}.
$$
Choose a unit vector $w$ orthogonal to both $\vec b$ and $\vec c$. Such a
vector exists because two vectors in $\mathbb R^3$ span at most a plane. Then
$N_\eta w=\eta w$ and $N_\eta^T w=\eta w$. Hence
$N_\eta^TN_\eta w=\eta^2w$, so one singular value of $N_\eta$, and therefore
of $M_\eta$, is $\eta$ for $\eta\ge0$. Define
$$
h:=|\vec a|\,|\vec b|,
\qquad
m_{\PhiP}:=\vec a^{\,T}D_{\PhiP}\vec b=\vec c\!\cdot\!\vec b .
$$
Let $s_+$ and $s_-$ be the other two singular values of $M_\eta$. These
symbols do not impose an ordering. When an ordered list is needed, we write
$s_1\ge s_2\ge s_3$. The squared Frobenius norm is the sum of the squared
singular values:
$$
\|N_\eta\|_F^2=\eta^2+s_+^2+s_-^2.
$$
A direct calculation gives
$$
\|N_\eta\|_F^2
=3\eta^2+2\eta\,\vec c\!\cdot\!\vec b+|\vec c|^2|\vec b|^2
=3\eta^2+2\eta m_{\PhiP}+h^2,
$$
and therefore
$$
s_+^2+s_-^2=2\eta^2+2\eta m_{\PhiP}+h^2.
$$
The product $s_+s_-$ comes from the determinant. With
$\det D_{\PhiP}=-1$ and $N_\eta=D_{\PhiP}M_\eta$, $\det M_\eta=-\det N_\eta$.
For vectors $u,v\in\mathbb R^3$,
$$
\det(\eta I+uv^T)=\eta^2(\eta+v^Tu).
$$
Applying this identity to $N_\eta=\eta I+\vec c\,\vec b^{\,T}$ gives
$$
\det M_\eta
=-\eta^2(\eta+\vec b^{\,T}\vec c)
=-\eta^2(\eta+m_{\PhiP}).
$$
Therefore the three singular values multiply to
$|\det M_\eta|=\eta^2|\eta+m_{\PhiP}|$. The noise point $\eta=0$
($\rho_0=\sigma$) is separable, hence below threshold. For $\eta>0$,
dividing by the identified value $\eta$,
$$
s_+s_-=\eta|\eta+m_{\PhiP}|.
$$
With $s_+^2+s_-^2=2\eta^2+2\eta m_{\PhiP}+h^2$, the sum
$(s_++s_-)^2=s_+^2+s_-^2+2\eta|\eta+m_{\PhiP}|$ is
$$
(s_++s_-)^2=
\begin{cases}
h^2+4\eta(\eta+m_{\PhiP}), & \eta+m_{\PhiP}\ge0,\\
h^2, & \eta+m_{\PhiP}<0.
\end{cases}
$$
The determinant
$$
\det T_{\rho_\eta}
=\frac{\det M_\eta}{(1+\eta)^3}
=-\frac{\eta^2(\eta+m_{\PhiP})}{(1+\eta)^3}
$$
is negative when $\eta+m_{\PhiP}>0$. In that case
\eqref{eq:supp-horodecki-f-singvals} uses the sum of all three singular
values, and therefore
$$
f(\rho_\eta)=\frac14\left(1+
\frac{\eta+\sqrt{h^2+4\eta(\eta+m_{\PhiP})}}{1+\eta}\right).
$$
The teleportation threshold condition $f(\rho_\eta)=\tfrac12$ is therefore
\begin{equation}
h^2+4\eta(\eta+m_{\PhiP})=1.
\label{eq:supp-product-tele-threshold-eq}
\end{equation}
For a given product noise state, the nonnegative solution of
\eqref{eq:supp-product-tele-threshold-eq} is
\begin{equation}
\eta_F
=\frac{\sqrt{1-h^2+m_{\PhiP}^2}-m_{\PhiP}}{2},
\label{eq:supp-EACF-product}
\end{equation}
We now check that this solution is the threshold. Indeed,
$$
\eta_F+m_{\PhiP}
=\frac{\sqrt{1-h^2+m_{\PhiP}^2}+m_{\PhiP}}{2}\ge0,
$$
because $h\le1$. The inequality is strict unless one is on the pure-product
boundary $h=1$ with $m_{\PhiP}\le0$. Hence the threshold value always lies
in the case $\eta+m_{\PhiP}\ge0$. In the complementary case
$\eta+m_{\PhiP}<0$ the determinant is positive and the Horodecki formula uses
the branch
$$
f(\rho_\eta)=\frac14\left(1+\frac{s_1+s_2-s_3}{1+\eta}\right),
$$
where $s_1\ge s_2\ge s_3$ are the singular values of $M_\eta$. The case
formula above gives
$s_++s_-=h$, and
$s_1+s_2-s_3\le s_1+s_2+s_3=\eta+s_++s_-$, so the
positive-determinant branch obeys
$$
f(\rho_\eta)
\le \frac14\left(1+\frac{\eta+h}{1+\eta}\right)
\le \frac12,
$$
because $h=|\vec a|\,|\vec b|\le1$. Therefore no crossing occurs in the
case $\eta+m_{\PhiP}<0$. This validates
\eqref{eq:supp-product-tele-threshold-eq} as the equation determining the
threshold, and hence
$\EACF(A\otimes B;\PhiP)=\eta_F$. For fixed local radii, the only
orientation-dependent quantity
in this equation is $m_{\PhiP}=\vec a^{\,T}D_{\PhiP}\vec b$. The matrix
$D_{\PhiP}$ is not a multiple of the identity, so $m_{\PhiP}$ is fixed by
the orientation of $A$ and $B$ relative to the $\PhiP$ frame, not by their
orientation to each other alone. With $\lambda_F=\EACF/(1+\EACF)$, the
teleportation threshold therefore depends not only on the local radii but
also on this Bell-frame alignment. By contrast,
$\EAC(A\otimes B)=\sqrt{(1-\Tr A^2)(1-\Tr B^2)}$, which depends only on the
local radii $|\vec a|,|\vec b|$.

For another Bell reference $\beta$, replace $D_{\PhiP}$ by the corresponding
Bell correlation matrix $D_\beta$ and set
$m_\beta:=\vec a^{\,T}D_\beta\vec b$. The same calculation gives
$$
\EACF(A\otimes B;\beta)
=\frac{\sqrt{1-h^2+m_\beta^2}-m_\beta}{2}.
$$
Explicitly,
$$
D_{\PhiP}=\operatorname{diag}(1,-1,1),\quad
D_{\PhiM}=\operatorname{diag}(-1,1,1),
$$
$$
D_{\PsiP}=\operatorname{diag}(1,1,-1),\quad
D_{\PsiM}=\operatorname{diag}(-1,-1,-1).
$$

\paragraph*{Arbitrary pure-state reference.}
For a general pure-state reference the Bell-frame quadratic equation is replaced by
a quartic equation for the crossing. For product noise $\sigma=A\otimes B$ with
$A=(I+\vec a\cdot\vec\sigma)/2$ and
$B=(I+\vec b\cdot\vec\sigma)/2$, we work in the Schmidt frame of
$|\psi\rangle$. In this frame the reference correlation tensor is represented
by the matrix $K_\psi=\operatorname{diag}(C_\psi,-C_\psi,1)$. Along
$\rho_\eta=(\sigma+\eta|\psi\rangle\!\langle\psi|)/(1+\eta)$, the
correlation tensor entering the Horodecki formula is represented by
$$
T_{\rho_\eta}=\frac{M_\psi(\eta)}{L},\qquad
M_\psi(\eta):=\eta K_\psi+\vec a\,\vec b^{\,T},\qquad
L:=1+\eta,
$$
and write
$$
d(\eta):=\det M_\psi(\eta).
$$
The sign of $d(\eta)$ selects the branch of the Horodecki formula. The same
pattern appeared in the Bell-reference calculation above, where
$\det M_\eta=-\eta^2(\eta+m_{\PhiP})$ has its sign set by the factor
$\eta+m_{\PhiP}$, the prefactor $-\eta^2$ being negative. Accordingly $d(\eta)<0$ is the analogue of
$\eta+m_{\PhiP}>0$. The branch $d(\eta)>0$, the analogue of
$\eta+m_{\PhiP}<0$, gives no teleportation crossing. For a general pure reference, the
calculation is more technical, so the two cases must be treated separately.
When $d(\eta)=0$ one singular value vanishes and the two Horodecki formulas
agree, so we include this case with $d(\eta)\le0$ in the final threshold
statement.

\begin{theorem}[Pure-reference product-noise teleportation threshold]
\label{thm:pure-reference-product-teleportation}
For product noise $\sigma=A\otimes B$ and a pure reference $|\psi\rangle$
with $C_\psi>0$, use the Schmidt frame of $|\psi\rangle$ and the notation
above. Define
$$
\begin{gathered}
h:=|\vec a|\,|\vec b|,\qquad
m:=\vec a^{\,T}K_\psi\vec b,\qquad
\mu:=\vec b^{\,T}K_\psi^{-1}\vec a,\\
k_a:=\vec a^{\,T}K_\psi^2\vec a,\qquad
k_b:=\vec b^{\,T}K_\psi^2\vec b,\qquad
w:=|\vec a|^2k_b+|\vec b|^2k_a .
\end{gathered}
$$
Then the fidelity absorption capacity $\EACF(A\otimes B;\psi)$ is obtained
from the nonnegative roots of the quartic equation
\begin{equation}
\boxed{\sum_{k=0}^{4}A_k\,\eta^k=0},
\label{eq:EACF-product-pure-quartic}
\end{equation}
where
$$
\boxed{
\begin{gathered}
A_4=-16C_\psi^2,\qquad
A_3=-16C_\psi^2(1+\mu),\\
A_2=4(1-C_\psi^2)-8(m+C_\psi^2\mu)
-4(1+C_\psi^2)h^2+4w,\\
A_1=4(1-m)(1-h^2),\qquad
A_0=(1-h^2)^2.
\end{gathered}
}
$$
Among the nonnegative roots of \eqref{eq:EACF-product-pure-quartic}, keep the
roots that satisfy $f(\rho_\eta)=1/2$ when substituted into the original
mixing line. The capacity $\EACF(A\otimes B;\psi)$ is the smallest of these
roots. The selected root satisfies $d(\eta)\le0$. When $C_\psi=1$, the equation reduces
to the Bell-reference formula
$$
\EACF(A\otimes B;\PhiP)
=\frac{\sqrt{1-h^2+m_{\PhiP}^2}-m_{\PhiP}}{2}.
$$
\end{theorem}

The calculation is organized as follows.
\begin{enumerate}
\item On the $d(\eta)<0$ branch, derive
Eq.~\eqref{eq:EACF-product-pure-algebraic} for $\eta$ from the condition
$f(\rho_\eta)=1/2$, then substitute the explicit polynomials $\nu_1$,
$\nu_2$, and $d(\eta)$ in $\eta$ to obtain the
quartic equation~\eqref{eq:EACF-product-pure-quartic}.
\item Check that $C_\psi=1$ recovers the Bell-reference formula.
\item Show that $d(\eta)>0$ gives no teleportation threshold.
\end{enumerate}
On the $d(\eta)<0$ branch
of \eqref{eq:supp-horodecki-f-singvals} the fully entangled fraction is
$f(\rho_\eta)=\tfrac14\bigl(1+(s_1+s_2+s_3)/L\bigr)$, with
$s_1\ge s_2\ge s_3\ge0$ the singular values of $M_\psi$, so
$f(\rho_\eta)=\tfrac12$ is the sum-of-singular-values condition
$$
s_1+s_2+s_3=L.
$$

We now rewrite this condition using the following quantities from
$S:=M_\psi M_\psi^T$ (the entries of $M_\psi$ are real, so
$M_\psi^T=M_\psi^\dagger$). By definition of the singular values, the
eigenvalues of $S$ are $s_1^2,s_2^2,s_3^2$, so
$$
\nu_1:=\operatorname{Tr}S=\sum_is_i^2,\quad
\nu_2:=\tfrac12\bigl[(\operatorname{Tr}S)^2-\operatorname{Tr}S^2\bigr]
=\sum_{i<j}s_i^2s_j^2,\quad
\det S=d(\eta)^2.
$$
Let $p_1=\sum_is_i$, $p_2=\sum_{i<j}s_is_j$, and
$p_3=s_1s_2s_3$ be the elementary symmetric polynomials of
the singular values. On the $d(\eta)<0$ branch, $|\det M_\psi|=-d(\eta)$, so they
satisfy
\begin{subequations}\label{eq:sym-power-relations}
\begin{align}
p_1^2&=\nu_1+2p_2,\label{eq:sym-rel-a}\\
p_2^2&=\nu_2+2p_1p_3,\label{eq:sym-rel-b}\\
p_3&=|\det M_\psi|=-d(\eta),\label{eq:sym-rel-c}
\end{align}
\end{subequations}
This is the Horodecki branch where the formula uses
$s_1+s_2+s_3$, so the crossing condition is $p_1=L$. The
calculation below derives the equation for possible crossings in this case. We
discuss at the end why the $d(\eta)>0$ branch does not give this crossing
condition. Equation~\eqref{eq:sym-rel-a} gives
$p_2=(L^2-\nu_1)/2$, and substituting into \eqref{eq:sym-rel-b} with
$p_3=-d(\eta)$ eliminates the singular values from the crossing condition,
\begin{equation}
(L^2-\nu_1)^2=4\nu_2-8L\,d(\eta).
\label{eq:EACF-product-pure-algebraic}
\end{equation}
This equation uses only the Frobenius norm
$\nu_1=\|M_\psi\|_F^2$, the quantity
$\nu_2=\sum_{i<j}s_i^2s_j^2$, and the determinant $d(\eta)$. In terms of the
quantities $h$, $m$, $\mu$, $k_a$, $k_b$, and $w$ defined in
Theorem~\ref{thm:pure-reference-product-teleportation}, and with
$L=1+\eta$, the three functions $\nu_1$, $\nu_2$, and $d$ are explicit
polynomials in $\eta$, which we evaluate below. Substituting them into
\eqref{eq:EACF-product-pure-algebraic} gives the quartic equation
\eqref{eq:EACF-product-pure-quartic}, with the coefficients stated in
Theorem~\ref{thm:pure-reference-product-teleportation}. Only the four
numbers $h$, $m$, $\mu$, $w$ enter, together with the reference
concurrence $C_\psi$. We now verify those coefficients. Define
$$
\tau:=\operatorname{Tr}K_\psi^2=1+2C_\psi^2,\qquad
\kappa:=\tfrac12\bigl[(\operatorname{Tr}K_\psi^2)^2-\operatorname{Tr}K_\psi^4\bigr]
=C_\psi^2(C_\psi^2+2),
$$
and $m_3:=\vec a^{\,T}K_\psi^3\vec b$. From
$$
S=M_\psi M_\psi^T
=\eta^2K_\psi^2+\eta\bigl(K_\psi\vec b\,\vec a^{\,T}
+\vec a\,\vec b^{\,T}K_\psi\bigr)+|\vec b|^2\vec a\,\vec a^{\,T},
$$
the trace is
$$
\begin{aligned}
\nu_1=\operatorname{Tr}S
&=\eta^2\operatorname{Tr}K_\psi^2
+\eta\operatorname{Tr}\bigl(K_\psi\vec b\,\vec a^{\,T}
+\vec a\,\vec b^{\,T}K_\psi\bigr)
+|\vec b|^2\operatorname{Tr}(\vec a\,\vec a^{\,T})\\
&=\tau\eta^2+2\eta\,\vec a^{\,T}K_\psi\vec b+|\vec a|^2|\vec b|^2\\
&=\tau\eta^2+2m\eta+h^2 .
\end{aligned}
$$

Similarly, since $\vec a\,\vec b^{\,T}$ has rank one, the matrix determinant
lemma $\det(A+\vec u\,\vec v^{\,T})=\det A\,(1+\vec v^{\,T}A^{-1}\vec u)$
gives, for $C_\psi>0$,
$$
\begin{aligned}
d(\eta)=\det(\eta K_\psi+\vec a\,\vec b^{\,T})
&=\det(\eta K_\psi)\bigl(1+\vec b^{\,T}(\eta K_\psi)^{-1}\vec a\bigr)\\
&=\eta^3\det K_\psi\bigl(1+\eta^{-1}\vec b^{\,T}K_\psi^{-1}\vec a\bigr)\\
&=-C_\psi^2\eta^2(\eta+\mu).
\end{aligned}
$$

It remains to compute $\nu_2=\tfrac12(\nu_1^2-\operatorname{Tr}S^2)$.
Let $\vec x:=K_\psi\vec b$. Then the middle term in $S$ is
$\eta(\vec x\,\vec a^{\,T}+\vec a\,\vec x^{\,T})$. Using
$\operatorname{Tr}(uv^T)=v^Tu$ and
$\operatorname{Tr}(uv^Txy^T)=(v^Tx)(y^Tu)$, the terms in
$\operatorname{Tr}S^2$ are
$$
\begin{aligned}
\operatorname{Tr}(\eta^2K_\psi^2)^2
&=\operatorname{Tr}K_\psi^4\,\eta^4,\\
2\operatorname{Tr}\!\left[
\eta^2K_\psi^2\cdot\eta(\vec x\,\vec a^{\,T}+\vec a\,\vec x^{\,T})
\right]
&=4m_3\eta^3,\\
\operatorname{Tr}\!\left[\eta^2(\vec x\,\vec a^{\,T}+\vec a\,\vec x^{\,T})^2\right]
&=2(m^2+|\vec a|^2k_b)\eta^2,\\
2\operatorname{Tr}\!\left[
\eta^2K_\psi^2\cdot|\vec b|^2\vec a\,\vec a^{\,T}
\right]
&=2|\vec b|^2k_a\eta^2,\\
2\operatorname{Tr}\!\left[
\eta(\vec x\,\vec a^{\,T}+\vec a\,\vec x^{\,T})
\cdot|\vec b|^2\vec a\,\vec a^{\,T}
\right]
&=4mh^2\eta,\\
\operatorname{Tr}\bigl(|\vec b|^2\vec a\,\vec a^{\,T}\bigr)^2
&=h^4.
\end{aligned}
$$
Therefore
$$
\operatorname{Tr}S^2
=\operatorname{Tr}K_\psi^4\,\eta^4
+4m_3\eta^3
+2\bigl(m^2+|\vec a|^2k_b+|\vec b|^2k_a\bigr)\eta^2
+4mh^2\eta+h^4.
$$
Squaring $\nu_1$ and subtracting $\operatorname{Tr}S^2$ gives
$$
\nu_2=\kappa\eta^4+2(\tau m-m_3)\eta^3
+(\tau h^2+m^2-w)\eta^2.
$$

Everything entering \eqref{eq:EACF-product-pure-algebraic} is now
available:
$$
\begin{aligned}
L&=1+\eta,\\
\nu_1&=\tau\eta^2+2m\eta+h^2,\\
d(\eta)&=-C_\psi^2\eta^2(\eta+\mu),\\
\nu_2&=\kappa\eta^4+2(\tau m-m_3)\eta^3+(\tau h^2+m^2-w)\eta^2.
\end{aligned}
$$
We also need
$$
\begin{aligned}
L^2-\nu_1
&=-2C_\psi^2\eta^2+2(1-m)\eta+(1-h^2),\\
8L\,d(\eta)
&=-8C_\psi^2\bigl[\eta^4+(1+\mu)\eta^3+\mu\eta^2\bigr],
\end{aligned}
$$
where the first line uses $1-\tau=-2C_\psi^2$. Expanding the square gives
$$
\begin{aligned}
(L^2-\nu_1)^2
&=4C_\psi^4\eta^4
-8C_\psi^2(1-m)\eta^3\\
&\quad
+\bigl[4(1-m)^2-4C_\psi^2(1-h^2)\bigr]\eta^2\\
&\quad
+4(1-m)(1-h^2)\eta+(1-h^2)^2.
\end{aligned}
$$
Substitute these expressions into
$(L^2-\nu_1)^2-4\nu_2+8L\,d(\eta)=0$ and collect powers of $\eta$.
The coefficient of $\eta^4$ is
$$
4C_\psi^4-4\kappa-8C_\psi^2=-16C_\psi^2.
$$
The coefficient of $\eta^3$ is
$$
-8C_\psi^2(1-m)-8(\tau m-m_3)-8C_\psi^2(1+\mu)
=-8\bigl[2C_\psi^2+(1+C_\psi^2)m-m_3+C_\psi^2\mu\bigr],
$$
where $\tau=1+2C_\psi^2$. The Cayley-Hamilton relation
$K_\psi^3-K_\psi^2-C_\psi^2K_\psi+C_\psi^2I=0$ follows from the
characteristic polynomial $(x-C_\psi)(x+C_\psi)(x-1)$. Equivalently,
$$
K_\psi^3=K_\psi^2+C_\psi^2K_\psi-C_\psi^2I.
$$
Rearranging gives
$$
(1+C_\psi^2)K_\psi-K_\psi^3=K_\psi-K_\psi^2+C_\psi^2I.
$$
Multiplying the Cayley-Hamilton identity by $K_\psi^{-1}$ gives
$$
K_\psi-K_\psi^2+C_\psi^2I=C_\psi^2K_\psi^{-1}.
$$
Therefore
$$
(1+C_\psi^2)K_\psi-K_\psi^3=C_\psi^2K_\psi^{-1}.
$$
Contracting with $\vec a^{\,T}$ and $\vec b$, and using
$\vec a^{\,T}K_\psi^{-1}\vec b=\vec b^{\,T}K_\psi^{-1}\vec a=\mu$, gives
$(1+C_\psi^2)m-m_3=C_\psi^2\mu$. Hence the cubic coefficient becomes
$$
-16C_\psi^2(1+\mu).
$$
The coefficient of $\eta^2$ is
$$
\begin{aligned}
&\bigl[4(1-m)^2-4C_\psi^2(1-h^2)\bigr]
-4(\tau h^2+m^2-w)-8C_\psi^2\mu\\
&=4(1-C_\psi^2)-8m-8C_\psi^2\mu
+(4C_\psi^2-4\tau)h^2+4w\\
&=4(1-C_\psi^2)-8(m+C_\psi^2\mu)
-4(1+C_\psi^2)h^2+4w.
\end{aligned}
$$
The last step uses $\tau=1+2C_\psi^2$. Finally, only
$(L^2-\nu_1)^2$ contributes to the coefficients of $\eta$ and $\eta^0$,
so
$$
A_1=4(1-m)(1-h^2),\qquad A_0=(1-h^2)^2.
$$
These are exactly the coefficients in
\eqref{eq:EACF-product-pure-quartic}.

At $C_\psi=1$ one has $K_\psi=D_{\PhiP}$, $\mu=m=m_{\PhiP}$, and $w=2h^2$,
and the left-hand side of \eqref{eq:EACF-product-pure-quartic} factors as
$$
-\bigl(4\eta^2+4m_{\PhiP}\eta-(1-h^2)\bigr)
\bigl(4\eta^2+4\eta+(1-h^2)\bigr).
$$
The first factor is exactly the Bell-reference equation
\eqref{eq:supp-product-tele-threshold-eq}, so its nonnegative root is
\eqref{eq:supp-EACF-product}. The second factor has only the roots
$\eta=(-1\pm h)/2\le0$, so it does not give an additional positive crossing
(when $h=1$, its zero root is the starting point). For
$C_\psi\neq1$ the quartic equation does not factor so simply, and a positive
root of \eqref{eq:EACF-product-pure-quartic} need not lie on the
$d(\eta)<0$ branch from which the equation was derived. We now show that the
teleportation threshold always lies on the $d(\eta)\le0$ branch, so that
\eqref{eq:EACF-product-pure-quartic} can be used to compute it for every
orientation.

\begin{lemma}[Correlation matrices of maximally entangled two-qubit states]
\label{lem:max-ent-correlation}
A real $3\times3$ matrix $R$ is the correlation matrix
$R_{ij}=\langle\beta|\sigma_i\otimes\sigma_j|\beta\rangle$ of some maximally
entangled two-qubit state $|\beta\rangle$ if and only if $R$ is orthogonal
with $\det R=-1$.
\end{lemma}

\begin{proof}
Every maximally entangled state is $|\beta\rangle=(U\otimes V)|\PhiP\rangle$
for some local unitaries $U,V$. The traceless Hermitian $2\times2$ matrices
form a real vector space with inner product
$\langle A,B\rangle=\tfrac12\Tr(AB)$, in which $\sigma_1,\sigma_2,\sigma_3$
are an orthonormal basis. Conjugation $A\mapsto UAU^\dagger$ maps this space
to itself and preserves the inner product, so it is an orthogonal
transformation. Its matrix $O_U$ in the Pauli basis, given by
$U\sigma_i U^\dagger=\sum_k(O_U)_{ki}\sigma_k$, is therefore orthogonal, and
$\det O_U=+1$ by continuity from $U=I$, so $O_U\in SO(3)$. The same holds for
$V$. Moving the local unitaries onto the Pauli operators and using
$U^\dagger\sigma_i U=\sum_k(O_U)_{ik}\sigma_k$ (and likewise for $V$),
$$
\begin{aligned}
R_{ij}&=\Tr\bigl[(U\otimes V)|\PhiP\rangle\!\langle\PhiP|(U\otimes V)^\dagger
\,\sigma_i\otimes\sigma_j\bigr]\\
&=\Tr\bigl[|\PhiP\rangle\!\langle\PhiP|\,
(U^\dagger\sigma_i U)\otimes(V^\dagger\sigma_j V)\bigr]\\
&=\sum_{k,l}(O_U)_{ik}(T_{\PhiP})_{kl}(O_V)_{jl}
=(O_U\,T_{\PhiP}\,O_V^{\mathsf T})_{ij},
\end{aligned}
$$
where $T_{\PhiP}=\operatorname{diag}(1,-1,1)$ is the Bell correlation matrix.
The matrix $R=O_U\operatorname{diag}(1,-1,1)\,O_V^{\mathsf T}$ is orthogonal with
$\det R=(+1)(-1)(+1)=-1$. Conversely, take any orthogonal $R$ with
$\det R=-1$. The product $O:=R\operatorname{diag}(1,-1,1)$ is orthogonal, and
$\det O=\det R\,\det\operatorname{diag}(1,-1,1)=(-1)(-1)=+1$, so $O\in SO(3)$.
Since $\operatorname{diag}(1,-1,1)^2=I$, we have
$O\operatorname{diag}(1,-1,1)=R$. The map $U\mapsto O_U$ from $SU(2)$ onto
$SO(3)$ is surjective, so $O=O_U$ for some single-qubit unitary $U$. Taking
$V=I$, the state $|\beta\rangle=(U\otimes I)|\PhiP\rangle$ then has
correlation matrix $O_U\operatorname{diag}(1,-1,1)=R$ by the forward
computation.
\end{proof}

To find when $f(\rho_\eta)$ exceeds $\tfrac12$, we recast the fully entangled
fraction as a maximization over maximally entangled states. It is the largest overlap of
$\rho_\eta$ with a maximally entangled two-qubit state $|\beta\rangle$. Such a
state has maximally mixed marginals, so the local terms $\sigma_i\otimes I$
and $I\otimes\sigma_j$ are absent from its Pauli expansion, leaving
$|\beta\rangle\!\langle\beta|
=\tfrac14\bigl(I\otimes I+\sum_{ij}R_{ij}\,\sigma_i\otimes\sigma_j\bigr)$,
with correlation matrix
$R_{ij}=\Tr(|\beta\rangle\!\langle\beta|\,\sigma_i\otimes\sigma_j)$. By
Lemma~\ref{lem:max-ent-correlation}, $R$ runs over exactly the orthogonal
$3\times3$ matrices of determinant $-1$. Expanding $\rho_\eta$ in the same
basis, the orthogonality $\Tr[(\sigma_i\otimes\sigma_j)(\sigma_k\otimes
\sigma_l)]=4\delta_{ik}\delta_{jl}$ makes the local terms of
$\rho_\eta$ drop out, and the overlap reduces to the correlation tensors
alone,
$$
\langle\beta|\rho_\eta|\beta\rangle
=\frac14\Bigl(1+\sum_{ij}(T_{\rho_\eta})_{ij}R_{ij}\Bigr)
=\frac14\Bigl(1+\frac{\Tr(M_\psi(\eta)R^{\mathsf T})}{L}\Bigr),
$$
with $T_{\rho_\eta}=M_\psi(\eta)/L$ the correlation tensor of $\rho_\eta$.
Maximizing the second term over $R$ (the transpose drops out, as $R^{\mathsf T}$
ranges over the same set) gives
$$
\Theta(\eta):=\max_{R\in O(3),\,\det R=-1}
\operatorname{Tr}\bigl(M_\psi(\eta)R\bigr),
\qquad
f(\rho_\eta)=\frac14\Bigl(1+\frac{\Theta(\eta)}{L}\Bigr).
$$
Evaluating the maximum over $R$ gives
$\Theta=s_1+s_2+s_3$ when $\det M_\psi<0$ and
$\Theta=s_1+s_2-s_3$ when $\det M_\psi\ge0$, which is exactly
\eqref{eq:supp-horodecki-f-singvals}. For each fixed $R$,
$\operatorname{Tr}(M_\psi(\eta)R)$ is affine in $\eta$. Taking the maximum over
$R$ gives a convex function, so $\Theta$ is convex, and so is
$$
\Delta(\eta):=\Theta(\eta)-(1+\eta).
$$
This quantity measures the difference between the optimized Horodecki
numerator and the value needed to reach $f=\tfrac12$, because
$$
f(\rho_\eta)-\frac12=\frac{\Delta(\eta)}{4L}.
$$
Since $L=1+\eta>0$, $f(\rho_\eta)-\tfrac12$ has the same sign as
$\Delta(\eta)$.
At $\eta=0$, $M_\psi(0)=\vec a\,\vec b^{\,T}$ is a rank-one outer product.
Its only nonzero singular value is $|\vec a|\,|\vec b|=h$, so
$\Theta(0)=h$ and $\Delta(0)=h-1\le0$.

\begin{proposition}[The threshold has $d(\eta)\le0$]
\label{prop:threshold-branch}
For every product noise state $\sigma=A\otimes B$ and every pure-state
reference $|\psi\rangle$ with $C_\psi>0$, let
$d(\eta):=\det M_\psi(\eta)$. Then $f(\rho_\eta)\le\tfrac12$ whenever
$d(\eta)>0$. Hence the threshold value $\eta=\EACF(\sigma;\psi)$ satisfies
$d(\eta)\le0$. For $d(\eta)\le0$, the crossing condition is the one that led to
\eqref{eq:EACF-product-pure-quartic}.
\end{proposition}

\begin{proof}
Recall $d(\eta)=-C_\psi^2\eta^2(\eta+\mu)$ with
$\mu=\vec b^{\,T}K_\psi^{-1}\vec a$, so $d(\eta)>0$ requires $\mu<0$ and then
holds exactly for $\eta\in(0,-\mu)$. If $\mu\ge0$ there is nothing to prove.
Assume $\mu<0$ and put $\eta_0=-\mu$. At $\eta_0$, $\det M_\psi(\eta_0)=0$,
so the smallest singular value vanishes and
$\Theta(\eta_0)=s_1+s_2$.
Lemma~\ref{lem:singular-sum-bound} gives $s_1+s_2\le1+\eta_0$, that is
$\Delta(\eta_0)\le0$. As $\Delta$ is convex with $\Delta(0)\le0$ and
$\Delta(\eta_0)\le0$, it is nonpositive on $[0,\eta_0]$, so
$f(\rho_\eta)\le\tfrac12$ for all $\eta$ with $d(\eta)>0$. The onset of
$f>\tfrac12$ therefore occurs at some $\eta\ge\eta_0$, where $d(\eta)\le0$. At that point
the fully entangled fraction uses $s_1+s_2+s_3$. The
threshold condition is therefore $s_1+s_2+s_3=L$. The
calculation above rewrites this condition as the polynomial
equation \eqref{eq:EACF-product-pure-quartic}.
\end{proof}

By Proposition~\ref{prop:threshold-branch}, $\EACF(\sigma;\psi)$ is the
smallest nonnegative root of \eqref{eq:EACF-product-pure-quartic} at which
$f(\rho_\eta)=\tfrac12$. The root check removes roots introduced by
squaring. In the Bell case, for instance, the extra factor
$4\eta^2+4\eta+(1-h^2)$ has no positive roots, and its possible zero root
at $h=1$ is just the starting point. In the Bell-reference
case and the structured subfamilies considered next the selected root satisfies
$d(\eta)\le0$, and the formulas reduce to the closed forms quoted in the main
text.

It remains to prove the singular-value bound used in
Proposition~\ref{prop:threshold-branch}.

\begin{lemma}[Singular-value bound at the determinant-zero point]
\label{lem:singular-sum-bound}
If $\mu=\vec b^{\,T}K_\psi^{-1}\vec a<0$ and $\eta_0=-\mu$, the
$3\times3$ matrix $M_\psi(\eta_0)=\eta_0 K_\psi+\vec a\,\vec b^{\,T}$ has
$\det M_\psi(\eta_0)=0$, and its two largest singular values satisfy
$s_1+s_2\le 1+\eta_0$, with equality only if
$|\vec a|=|\vec b|=1$.
\end{lemma}

\begin{proof}
The point of the proof is to control $s_1+s_2$ at the endpoint of the
positive-determinant interval. At $\eta_0=-\mu$ one has
$\det M_\psi(\eta_0)=0$, so $s_3=0$ and
$$
(s_1+s_2)^2=s_1^2+s_2^2+2s_1s_2.
$$
We compute the two terms on the right separately. The first is immediate:
the sum of all three squared singular values is the squared Frobenius norm,
and $s_3=0$ here, so
$$
s_1^2+s_2^2=\|M_\psi(\eta_0)\|_F^2
=(1+2C_\psi^2)\eta_0^2+2m\eta_0+h^2,\qquad
m=\vec a^{\,T}K_\psi\vec b,\quad
\eta_0=-\vec a^{\,T}K_\psi^{-1}\vec b .
$$

The product $s_1s_2$ takes more work. For an invertible $3\times3$ matrix
$X$, the matrix $(\det X)X^{-1}$ has entries that are polynomial functions of
the entries of $X$. Use $\Gamma(X)$ for that same polynomial matrix for every
$3\times3$ matrix $X$. It satisfies
$$
X\Gamma(X)=(\det X)I,
$$
and equals $(\det X)X^{-1}$ when $X$ is invertible. This matrix
helps because its singular values are the pairwise products of those of $X$.
Indeed, if $X$ has singular values $r_1,r_2,r_3$, then for invertible $X$ the
singular values of $X^{-1}$ are $1/r_1,1/r_2,1/r_3$, and multiplying by the
scalar $\det X$ rescales them by $|\det X|=r_1r_2r_3$, so $\Gamma(X)$ has
singular values $r_2r_3$, $r_1r_3$, $r_1r_2$. The singular case follows by
continuity. For
$M_\psi(\eta_0)$, with singular values $s_1,s_2,0$, only the product
$s_1s_2$ survives, so $\|\Gamma(M_\psi(\eta_0))\|_F=s_1s_2$ is exactly what
we need.

It remains to evaluate $\Gamma(M_\psi(\eta_0))$ in closed form. This is direct
because $M_\psi(\eta)=\eta K_\psi+\vec a\,\vec b^{\,T}$ adds the rank-one
matrix $\vec a\,\vec b^{\,T}$ to $\eta K_\psi$, and the inverse is explicit:
for $\eta\ne\eta_0$,
$$
M_\psi(\eta)^{-1}
=\frac{1}{\eta}K_\psi^{-1}
-\frac{1}{\eta(\eta+\mu)}
K_\psi^{-1}\vec a\,\vec b^{\,T}K_\psi^{-1}.
$$
Also,
$d(\eta)=\det M_\psi(\eta)=-C_\psi^2\eta^2(\eta+\mu)$, so
$$
d(\eta)M_\psi(\eta)^{-1}
=-C_\psi^2\eta(\eta+\mu)K_\psi^{-1}
+C_\psi^2\eta\,(K_\psi^{-1}\vec a)(K_\psi^{-1}\vec b)^T .
$$
Taking $\eta\to\eta_0=-\mu$ in this polynomial matrix, the first term
vanishes and
$$
\Gamma(M_\psi(\eta_0))
=C_\psi^2\eta_0\,(K_\psi^{-1}\vec a)(K_\psi^{-1}\vec b)^T.
$$
A rank-one matrix $c\,\vec u\,\vec v^{\,T}$ has Frobenius norm
$|c|\,|\vec u|\,|\vec v|$: indeed
$$
\|c\,\vec u\,\vec v^{\,T}\|_F^2=\sum_{i,j}(c\,u_iv_j)^2=c^2|\vec u|^2|\vec v|^2 .
$$
Hence
$$
s_1s_2=\|\Gamma(M_\psi(\eta_0))\|_F
=C_\psi^2\eta_0\,|K_\psi^{-1}\vec a|\,|K_\psi^{-1}\vec b|.
$$

Combining the two computed pieces gives
$$
(s_1+s_2)^2=(1+2C_\psi^2)\eta_0^2+2m\eta_0+h^2
+2C_\psi^2\eta_0\,|K_\psi^{-1}\vec a|\,|K_\psi^{-1}\vec b|.
$$
We want the sign of $(s_1+s_2)^2-(1+\eta_0)^2$. Subtracting
$(1+\eta_0)^2=1+2\eta_0+\eta_0^2$ from the line above and collecting the
terms proportional to $\eta_0$,
$$
(s_1+s_2)^2-(1+\eta_0)^2=2\eta_0\,(T-1)-(1-h^2),
$$
where
$$
T:=(C_\psi^2\eta_0+m)+C_\psi^2\,|K_\psi^{-1}\vec a|\,|K_\psi^{-1}\vec b| .
$$
In $T$ the term $C_\psi^2|K_\psi^{-1}\vec a|\,|K_\psi^{-1}\vec b|$ is already
a product of lengths. Only the leftover $C_\psi^2\eta_0+m$ still mixes
$K_\psi^{-1}$ (through $\eta_0=-\vec a^{\,T}K_\psi^{-1}\vec b$) and $K_\psi$
(through $m=\vec a^{\,T}K_\psi\vec b$). It collapses because
$K_\psi-C_\psi^2K_\psi^{-1}=\operatorname{diag}(0,0,1-C_\psi^2)$:
$$
C_\psi^2\eta_0+m
=\vec a^{\,T}(K_\psi-C_\psi^2K_\psi^{-1})\vec b
=(1-C_\psi^2)a_zb_z,
$$
so that
$$
T=(1-C_\psi^2)\,a_zb_z+C_\psi^2\,|K_\psi^{-1}\vec a|\,|K_\psi^{-1}\vec b| .
$$

Since $h=|\vec a|\,|\vec b|\le1$, the right-hand side is nonpositive as soon
as $T\le1$, which we now establish. Put $t:=1-C_\psi^2\in[0,1)$ (as $0<C_\psi\le1$). In coordinates,
$$
C_\psi^2|K_\psi^{-1}\vec a|^2=a_x^2+a_y^2+C_\psi^2a_z^2=|\vec a|^2-t\,a_z^2\le1-t\,a_z^2,
$$
and likewise $C_\psi^2|K_\psi^{-1}\vec b|^2\le1-t\,b_z^2$, both from
$|\vec a|,|\vec b|\le1$ (with both right-hand sides positive, as $t<1$).
Multiplying and taking square roots,
$$
C_\psi^2\,|K_\psi^{-1}\vec a|\,|K_\psi^{-1}\vec b|\le\sqrt{(1-t\,a_z^2)(1-t\,b_z^2)} .
$$
By the AM-GM inequality $\sqrt{PQ}\le\tfrac12(P+Q)$, the right-hand side is
at most $\tfrac12\bigl[(1-t\,a_z^2)+(1-t\,b_z^2)\bigr]=1-\tfrac t2(a_z^2+b_z^2)$,
so, with $T=t\,a_zb_z+C_\psi^2|K_\psi^{-1}\vec a|\,|K_\psi^{-1}\vec b|$,
$$
T\le1+t\,a_zb_z-\tfrac t2(a_z^2+b_z^2)=1-\tfrac t2(a_z-b_z)^2\le1 .
$$
Therefore $s_1+s_2\le1+\eta_0$. Equality forces $h=1$ (so
$|\vec a|=|\vec b|=1$) and $T=1$ (so $C_\psi=1$ or $a_z=b_z$).
\end{proof}

\subsection{Illustrative subfamily checks}
\label{subsec:supp-subfamily-checks}

This subsection illustrates the product law of
Theorem~\ref{thm:product} and Section~\ref{sec:eval1}: unlike $\lambda_*$,
the teleportation threshold $\lambda_F$ depends on the orientation of
$\sigma$ relative to the Bell frame. For a product state $\sigma=A\otimes B$ the
entanglement threshold is $\lambda_*=g/(2+g)$ with impurity factor
$g=\sqrt{(1-|\vec a|^2)(1-|\vec b|^2)}$
\eqref{eq:product-threshold}. This value is the same for every
orientation. We specialize \eqref{eq:supp-EACF-product} to two
orientations.

\paragraph*{Aligned-axis products.}
Take $\vec a=|\vec a|\hat z$ and $\vec b=|\vec b|\hat z$. Then
$h=m_{\PhiP}=|\vec a|\,|\vec b|$, so \eqref{eq:supp-EACF-product} gives
$\EACF=(1-|\vec a|\,|\vec b|)/2$ and
\begin{equation}
\lambda_F=\frac{1-|\vec a|\,|\vec b|}{3-|\vec a|\,|\vec b|}.
\label{eq:supp-lambdaF-aligned}
\end{equation}
Setting $\lambda_F=\lambda_*$ and cross-multiplying reduces to
$(|\vec a|-|\vec b|)^2=0$, so the two thresholds agree exactly on the
equal-radius diagonal $|\vec a|=|\vec b|=r$, where
\begin{equation}
\lambda_F(r,r)=\lambda_*(r,r)=\frac{1-r^2}{3-r^2},
\qquad r\in[0,1],
\label{eq:supp-equality-curve}
\end{equation}
interpolating from the Werner value $\lambda=1/3$ at $r=0$ to the
pure-product corner $\lambda=0$ at $r=1$.

\paragraph*{Perpendicular products.}
Take $\vec a=|\vec a|\hat z$ and $\vec b=|\vec b|\hat x$. Now the Bell-frame
alignment vanishes, $m_{\PhiP}=0$, while $h=|\vec a|\,|\vec b|$, so
\eqref{eq:supp-EACF-product} gives
$$
\lambda_F^{\perp}(\sigma)
=\frac{\sqrt{1-|\vec a|^2|\vec b|^2}}
{2+\sqrt{1-|\vec a|^2|\vec b|^2}}
\,\ge\,\lambda_*(\sigma).
$$
The inequality follows from the monotonicity of $x/(2+x)$ together with
$1-|\vec a|^2|\vec b|^2\ge(1-|\vec a|^2)(1-|\vec b|^2)$. Equality forces
$$
1-|\vec a|^2|\vec b|^2=(1-|\vec a|^2)(1-|\vec b|^2),
$$
that is $|\vec a|^2+|\vec b|^2=2|\vec a|^2|\vec b|^2$. With
$|\vec a|,|\vec b|\in[0,1]$, the only solutions are the maximally mixed
point $|\vec a|=|\vec b|=0$, where both thresholds equal $1/3$, and the
pure-product boundary point $|\vec a|=|\vec b|=1$ (so $h=1$), where both
thresholds vanish.

For perpendicular products, the equality set is only these two endpoints.
For aligned products, it is the whole equal-radius diagonal
\eqref{eq:supp-equality-curve}. The equality set $\lambda_F=\lambda_*$ on
the product family therefore depends on the orientation relative to the Bell
frame. On the $X$-state sector, by contrast, Theorem~\ref{thm:xstate-tele}
identifies the equality criterion as ``equal middle populations.''

\section{Channel evolution and unital monotonicity}
\label{sec:supp-channel-calculus}

This appendix records how the product formula for $\EAC$ changes under
three standard one-qubit channels: depolarizing, dephasing, and amplitude
damping. Since \eqref{eq:impurity} depends only on the local Bloch radii,
the channel calculation is obtained by substituting the evolved radii. The
same formulas give a monotonicity statement: local unital channels cannot
decrease $\EAC$ on the product family, while nonunital channels can.

\subsection{Local depolarizing channel}
\label{subsec:supp-depolarizing}

The single-qubit depolarizing channel
$\mathcal{D}_p(\rho)=(1-p)\rho+p\,I/2$ sends a Bloch vector $\vec a$
to $(1-p)\vec a$, leaving the maximally mixed state invariant. For a product
noise state $\sigma=A\otimes B$ with local
Bloch vectors $\vec a,\vec b$, applying
$\mathcal{D}_{p_A}\otimes\mathcal{D}_{p_B}$ yields another product state
with marginal radii $(1-p_A)|\vec a|$ and $(1-p_B)|\vec b|$.
The product law of Eq.~\eqref{eq:impurity} then gives
\begin{equation}
\begin{aligned}
&\EAC\!\bigl(\mathcal{D}_{p_A}(A)\otimes\mathcal{D}_{p_B}(B)\bigr)\\
&\quad=\tfrac{1}{2}\sqrt{
\bigl(1-(1-p_A)^2|\vec a|^2\bigr)
\bigl(1-(1-p_B)^2|\vec b|^2\bigr)}.
\end{aligned}
\label{eq:supp-EAC-depolarizing}
\end{equation}
At fixed input state and fixed value of the other channel parameter, the
corresponding factor is strictly increasing in $p_X$ on $[0,1]$ whenever
the input Bloch vector is nonzero. The capacity $\EAC$ is non-decreasing in each
depolarizing strength, and the increase is strict whenever the corresponding
marginal is not already maximally mixed. By the M\"obius transform
\eqref{eq:mobius}, the Bell-mixing threshold $\lambda_*$ is then a
non-decreasing function of every depolarizing parameter applied locally.

\subsection{Local dephasing channel}
\label{subsec:supp-dephasing}

The single-qubit dephasing channel
$\mathcal{Z}_p(\rho)=(1-p)\rho+p\,\sigma_z\rho\sigma_z$ acts on Bloch
components as $(a_x,a_y,a_z)\mapsto((1-2p)a_x,(1-2p)a_y,a_z)$.  It is
unital but anisotropic. It contracts the off-axis components while leaving
the $z$-axis component invariant.

For a product noise state with Bloch vectors of arbitrary orientation, the
post-channel marginal radii are
$$
\begin{aligned}
|\mathcal{Z}_{p_A}(\vec a)|^2&=(1-2p_A)^2(a_x^2+a_y^2)+a_z^2,\\
|\mathcal{Z}_{p_B}(\vec b)|^2&=(1-2p_B)^2(b_x^2+b_y^2)+b_z^2,
\end{aligned}
$$
and substituting these into the product law gives
\begin{equation}
\EAC\!\bigl(\mathcal{Z}_{p_A}(A)\otimes\mathcal{Z}_{p_B}(B)\bigr)
=\tfrac{1}{2}\sqrt{\bigl(1-|\mathcal{Z}_{p_A}\vec a|^2\bigr)\bigl(1-|\mathcal{Z}_{p_B}\vec b|^2\bigr)}.
\label{eq:supp-EAC-dephasing}
\end{equation}
Two special cases show the role of orientation. For aligned-$\hat z$ product
noise states ($\vec a,\vec b\parallel\hat z$),
$|\mathcal{Z}\vec a|=|\vec a|$, and $\EAC$ is unchanged by
local dephasing. For purely off-axis product noise states
($a_z=b_z=0$), local dephasing contracts both marginals isotropically
within the equatorial plane, and $\EAC$ increases monotonically up to
the Werner value $1/2$ as $p_A,p_B\to 1/2$ (the completely
dephasing limit, which for these off-axis marginals contracts the
equatorial Bloch vector fully to the origin, sending each marginal to $I/2$).

\subsection{Local amplitude-damping channel}
\label{subsec:supp-amplitude-damping}

The amplitude-damping channel $\mathcal{A}_\gamma$ with damping
strength $\gamma\in[0,1]$ has Kraus operators $K_0=\mathrm{diag}(1,\sqrt{1-\gamma})$
and $K_1=\sqrt\gamma\,|0\rangle\!\langle 1|$, sending Bloch components
to $(\sqrt{1-\gamma}\,a_x,\sqrt{1-\gamma}\,a_y,(1-\gamma)a_z+\gamma)$.
This channel is not unital: $\mathcal{A}_\gamma(I/2)
=\tfrac{1+\gamma}{2}|0\rangle\!\langle 0|+\tfrac{1-\gamma}{2}|1\rangle\!\langle 1|
\ne I/2$ for $\gamma>0$.

Take the aligned-$\hat z$ product noise state
$\vec a=|\vec a|\hat z$. Its marginal after amplitude damping
has Bloch vector $\vec a'=(0,0,(1-\gamma)|\vec a|+\gamma)$ with radius
$|\vec a'|=(1-\gamma)|\vec a|+\gamma$. This gives
$|\vec a'|\ge |\vec a|$ with strict inequality whenever
$\gamma>0$ and $|\vec a|<1$. Then
$$
1-|\vec a'|^2=1-\bigl((1-\gamma)|\vec a|+\gamma\bigr)^2,
$$
and an analogous formula on the $B$ side.  Substituting into the
product law,
\begin{equation}
\begin{aligned}
&\EAC\!\bigl(\mathcal{A}_{\gamma_A}(A)\otimes\mathcal{A}_{\gamma_B}(B)\bigr)\\
&\,=\tfrac{1}{2}\sqrt{
\bigl(1-((1-\gamma_A)|\vec a|+\gamma_A)^2\bigr)
\bigl(1-((1-\gamma_B)|\vec b|+\gamma_B)^2\bigr)},
\end{aligned}
\label{eq:supp-EAC-AD-aligned}
\end{equation}
which is closed-form on the aligned-$\hat z$ product family.

\paragraph*{Aligned-$+\hat z$ product noise states.}
For aligned-$+\hat z$ product noise, amplitude damping changes the local
radii to $(1-\gamma_A)|\vec a|+\gamma_A$ and
$(1-\gamma_B)|\vec b|+\gamma_B$. Each radius increases whenever the
corresponding input radius is below $1$. The factors $1-|\vec a|^2$ and
$1-|\vec b|^2$ in the product law then decrease, so $\EAC$ decreases
strictly whenever at least one nonpure marginal is damped and the other
side has nonzero impurity. As $\gamma_A,\gamma_B\to1$, both marginals
approach $|0\rangle\!\langle0|$ and $\EAC\to0$: the channel output is a
pure product, so no Bell weight can be added while staying separable.

\paragraph*{Negative-$z$ product noise states.}
For $\vec a=-r\hat z$ with $r>0$, amplitude damping gives radius
$|\gamma-(1-\gamma)r|$. This is smaller than $r$ for
$0<\gamma<2r/(1+r)$, as the channel moves the vector toward $+\hat z$
and through the origin. On this interval, the local impurity increases
and so can $\EAC$. After $\gamma>2r/(1+r)$ the radius grows again,
$\EAC$ then falls to $0$ as $\gamma\to1$ and the marginal becomes pure.
The initial direction of change depends on the sign of the Bloch
component, reflecting the nonunital shift of $\mathcal{A}_\gamma$.

\subsection{Monotonicity on products, and what fails beyond}
\label{subsec:supp-monotonicity}

The three channel calculations above all concern product noise states. On
that family, local unital noise can only raise $\EAC$.

\begin{theorem}[Channel monotonicity of $\EAC$ on products]
\label{thm:channel-monotonicity}
Let $\mathcal{N}_A,\mathcal{N}_B$ be any single-qubit
trace-preserving completely positive (CP) maps acting locally on
$A$ and $B$. If both maps are unital, then
$$
\EAC\!\bigl(\mathcal{N}_A(A)\otimes\mathcal{N}_B(B)\bigr)
\ge \EAC(A\otimes B).
$$
In Bloch-vector notation, the inequality is strict whenever one marginal
is strictly contracted, $|\mathcal{N}_A(\vec a)|<|\vec a|$, and the
other is not already pure, $|\vec b|<1$ (and symmetrically).
\end{theorem}

\paragraph*{Proof.}
The product law of Eq.~\eqref{eq:impurity} makes $\EAC$ depend on each
marginal only through its Bloch radius. Write a one-qubit state as
$\rho(\vec a)=(I+\vec a\cdot\vec\sigma)/2$. A trace-preserving qubit map
acts on Bloch vectors as an affine map $\vec a\mapsto T\vec a+\vec t$.
Unitality, $\mathcal{N}(I/2)=I/2$, is exactly $\vec t=0$. Since the map is
positive, it sends every density matrix to a density matrix. Equivalently,
$T$ sends the Bloch ball into itself. Hence $|T\vec u|\le1$ for every unit
vector $\vec u$. For a nonzero Bloch vector, write
$\vec a=|\vec a|\vec u$ with $|\vec u|=1$. Linearity gives
$|T\vec a|=|\vec a|\,|T\vec u|\le|\vec a|$. The same inequality is trivial
for $\vec a=0$~\cite{RuskaiSzarekWerner2002}. It follows that
$1-|T\vec a|^2\ge 1-|\vec a|^2$. Applying the product law
\eqref{eq:impurity} to the two local factors gives the inequality.
Strictness follows if one factor increases strictly and the other factor is
positive, for example if $|T_A\vec a|<|\vec a|$ and $|\vec b|<1$.
$\square$

\paragraph*{Non-unital counterexample.}
The amplitude-damping formula \eqref{eq:supp-EAC-AD-aligned} gives a
direct counterexample beyond unital channels. For aligned-$+\hat z$
product noise with $0<|\vec a|,|\vec b|<1$, local amplitude damping strictly
reduces $\EAC$. Hence the product-family monotonicity statement is genuinely
unital. It does not extend to arbitrary local channels.

\paragraph*{Scope.}
The theorem applies only to product noise. Non-product separable noise
states are outside this monotonicity statement.

\section{Evaluating the two capacities}
\label{subsec:evaluating-capacities}
This appendix records the general evaluation formulas used when no product or
$X$-state closed form is available.

\paragraph*{Entanglement absorption.}
The entanglement absorption capacity is the largest amount of Bell mixing
before $\sigma^{T_B}+\eta Q$ loses positivity. Here
$\sigma^{T_B}\succeq0$ on the separable set, while $Q=\Swap/2$ has a
one-dimensional negative eigenspace spanned by $|\PsiM\rangle$ with
eigenvalue $-\tfrac12$. Positivity can fail only along
directions where $Q$ has a negative expectation. Each such direction sets
its own bound: every $|x\rangle$ with $\langle x|Q|x\rangle<0$
forces $\eta\le\langle x|\sigma^{T_B}|x\rangle/(-\langle x|Q|x\rangle)$,
while directions with $\langle x|Q|x\rangle\ge0$ impose no bound. Taking
the tightest of these bounds gives the capacity,
\begin{equation}
\EAC(\sigma)=\min_{\langle x|Q|x\rangle<0}
\frac{\langle x|\sigma^{T_B}|x\rangle}{-\langle x|Q|x\rangle}.
\label{eq:EAC-Rayleigh}
\end{equation}
In the interior of the separable set, where $\sigma^{T_B}\succ0$, the same
quantity can be computed from an ordinary Hermitian eigenvalue problem. Set
$P:=\sigma^{T_B}$ and
$H:=P^{-1/2}QP^{-1/2}$. Then
$$
P+\eta Q
=P^{1/2}(I+\eta H)P^{1/2}.
$$
The matrix $P+\eta Q$ is positive semidefinite exactly when
$I+\eta H$ is positive semidefinite. If $\mu_{\min}<0$ is the smallest
eigenvalue of $H$, the first loss of positivity occurs at
$1+\eta\mu_{\min}=0$, hence
\begin{equation}
\EAC(\sigma)=-\frac{1}{\mu_{\min}},
\qquad
\mu_{\min}=\lambda_{\min}(H),
\label{eq:EAC-rescaled-eigenvalue}
\end{equation}
where $\lambda_{\min}(H)$ denotes the smallest eigenvalue of
$H=(\sigma^{T_B})^{-1/2}Q(\sigma^{T_B})^{-1/2}$.
Equivalently $\det(\sigma^{T_B}+\eta Q)=0$ locates the same crossing,
$\EAC$ being its first positive root. At the edge of the separable set
$\sigma^{T_B}$ can be singular, meaning that it has a zero eigenvalue. Then
$\det(\sigma^{T_B})=0$, so $\eta=0$ is automatically a root of
$\det(\sigma^{T_B}+\eta Q)=0$. This root only says that the starting point
lies on the boundary of the positive cone. It need not be the Bell-mixing
threshold. In this boundary case the square-root reduction above is also no
longer available, and \eqref{eq:EAC-Rayleigh} remains the correct
evaluation. The product, pure-reference, and $X$-state formulas of
Sections~\ref{sec:eval1} and~\ref{sec:eval2} are the cases where
$\sigma^{T_B}+\eta Q$ has enough structure for a closed-form solution.

\paragraph*{Fidelity absorption.}
For $\EACF$ the same parameter $\eta$ is used, but the condition is
$f(\rho_\eta)\le\tfrac12$ for
$\rho_\eta=(\sigma+\eta\PhiP)/(1+\eta)$. Expanding this inequality for
each maximally entangled $|\beta\rangle\in\mathcal M$ gives a witness
bound. For a fixed $|\beta\rangle$,
\begin{align*}
\langle\beta|\rho_\eta|\beta\rangle\le\frac12
&\Longleftrightarrow
\frac{\langle\beta|\sigma|\beta\rangle
+\eta\langle\beta|\PhiP|\beta\rangle}{1+\eta}
\le\frac12\\
&\Longleftrightarrow
\eta\Bigl(\langle\beta|\PhiP|\beta\rangle-\frac12\Bigr)
\le
\frac12-\langle\beta|\sigma|\beta\rangle .
\end{align*}
Since $\sigma$ is separable, $\langle\beta|\sigma|\beta\rangle\le1/2$.
Therefore a test state with $\langle\beta|\PhiP|\beta\rangle\le1/2$
imposes no upper bound on $\eta$. A test state with
$\langle\beta|\PhiP|\beta\rangle>1/2$ gives
$$
\eta\le
\frac{\tfrac12-\langle\beta|\sigma|\beta\rangle}
{\langle\beta|\PhiP|\beta\rangle-\tfrac12}.
$$
Taking the tightest such bound over all maximally entangled witnesses gives
\begin{equation}
\EACF(\sigma;\PhiP)=
\inf_{\substack{|\beta\rangle\in\mathcal M\\
\langle\beta|\PhiP|\beta\rangle>1/2}}
\frac{\tfrac12-\langle\beta|\sigma|\beta\rangle}
{\langle\beta|\PhiP|\beta\rangle-\tfrac12}.
\label{eq:EACF-witness-evaluation}
\end{equation}
This is the teleportation analogue of \eqref{eq:EAC-Rayleigh}: it finds
the first maximally entangled witness that reaches the level
$f(\rho_\eta)=\tfrac12$. The reductions to the closed forms can be seen
directly from this expression.

For product noise, let
$A=(I+\vec a\!\cdot\!\boldsymbol\sigma)/2$ and
$B=(I+\vec b\!\cdot\!\boldsymbol\sigma)/2$. If $R$ is the correlation
matrix of the maximally entangled witness $|\beta\rangle$, Lemma~\ref{lem:max-ent-correlation}
gives $R\in O(3)$ with $\det R=-1$. Since maximally entangled states have
maximally mixed one-qubit marginals, only the correlation tensors enter:
$$
\langle\beta|A\otimes B|\beta\rangle
=\frac14\left(1+\Tr(\vec a\,\vec b^{\,T}R^T)\right),
\qquad
\langle\beta|\PhiP|\beta\rangle
=\frac14\left(1+\Tr(D_{\PhiP}R^T)\right).
$$
Substituting these two overlaps in
\eqref{eq:EACF-witness-evaluation} gives
$$
\EACF(A\otimes B;\PhiP)
=\inf_{\substack{R\in O(3),\ \det R=-1\\
\Tr(D_{\PhiP}R^T)>1}}
\frac{1-\Tr(\vec a\,\vec b^{\,T}R^T)}
{\Tr(D_{\PhiP}R^T)-1}.
$$
The same reduction can be written as a feasibility condition. For a fixed
maximally entangled witness with correlation matrix $R$,
$$
\langle\beta|\rho_\eta|\beta\rangle
=\frac14\left(1+
\frac{\Tr((\vec a\,\vec b^{\,T}+\eta D_{\PhiP})R^T)}{1+\eta}\right).
$$
Since $1+\eta>0$, the inequality
$\langle\beta|\rho_\eta|\beta\rangle\le\tfrac12$ is equivalent to
$$
\Tr((\vec a\,\vec b^{\,T}+\eta D_{\PhiP})R^T)\le 1+\eta .
$$
Requiring this for every maximally entangled witness gives
$$
\max_{\substack{R\in O(3)\\ \det R=-1}}
\Tr\!\left((\vec a\,\vec b^{\,T}+\eta D_{\PhiP})R^T\right)
\le 1+\eta .
$$
Therefore $\EACF(A\otimes B;\PhiP)$ is the largest $\eta$ for which the
displayed inequality holds. The two sides are continuous in $\eta$, so the
largest feasible value is found by setting them equal.
This is precisely the Horodecki singular-value optimization carried out
in Appendix~\ref{subsec:supp-product-tele}. Solving it gives the
$\beta=\PhiP$ case of the Bell-frame product formula
\eqref{eq:EACF-product}.

For separable $X$ states, using \eqref{eq:EACF-witness-evaluation} is
equivalent to imposing $f(\rho_\eta)\le\tfrac12$. Section~\ref{subsec:xstate-fidelity}
evaluates this maximization directly on the $X$ family. In the
$|\PhiP\rangle$ frame the resulting crossing gives \eqref{eq:EAC-F-X}.

\section{Additional closed forms}
\label{sec:supp-extra-closed-forms}

The two main families are not the only noise states with closed-form
thresholds. This appendix records two additional closed forms: noise
diagonal in the computational basis, states with a rank-one correlation
tensor. Both follow by specializing the evaluation formulas in
Appendix~\ref{subsec:evaluating-capacities}.

\subsection{Computational-basis diagonal noise}
\label{subsec:extra-diagonal}
For a state diagonal in the computational basis,
$\sigma_{\rm diag}=\operatorname{diag}(a,b,c,d)$ with $a+b+c+d=1$,
the computation separates the two coordinates. In the $\Phi^\pm$ frames,
\begin{equation}
\EAC(\sigma_{\rm diag};\Phi^\pm)=2\sqrt{bc},
\qquad
\EACF(\sigma_{\rm diag};\Phi^\pm)=b+c,
\label{eq:extra-diag-Phi}
\end{equation}
and in the $\Psi^\pm$ frames,
\begin{equation}
\EAC(\sigma_{\rm diag};\Psi^\pm)=2\sqrt{ad},
\qquad
\EACF(\sigma_{\rm diag};\Psi^\pm)=a+d.
\label{eq:extra-diag-Psi}
\end{equation}
In each frame the entanglement threshold sees the geometric mean of the
opposite-block populations, while the teleportation threshold sees their
sum. On the diagonal
product slice
$A=(I+|\vec a|\sigma_z)/2$, $B=(I+|\vec b|\sigma_z)/2$, for which
Eq.~\eqref{eq:extra-diag-Phi} gives
$$
\EACF(A\otimes B;\Phi^\pm)=\frac{1-|\vec a|\,|\vec b|}{2},
\qquad
\lambda_F(A\otimes B;\Phi^\pm)
=
\frac{1-|\vec a|\,|\vec b|}{3-|\vec a|\,|\vec b|}.
$$

\subsection{Rank-one correlation tensors}
\label{subsec:extra-rank-one}
The product formula for $\EACF$ extends to a broader family. On the separable domain,
$\EACF$ depends on the noise state only through its correlation tensor
$(T_\sigma)_{ij}:=\Tr(\sigma\,\sigma_i\otimes\sigma_j)$. Any two noise
states with the same tensor share the same capacity. For every state whose
correlation tensor is rank one,
$$
T_\sigma=x\,y^T .
$$
Writing $h:=\|x\|\,\|y\|$ and the Bell-frame alignment $m_\beta:=x^TD_\beta y$,
where $D_\beta$ is the sign matrix \eqref{eq:Dbeta} of the reference
correlation tensor, the capacity is
\begin{equation}
\EACF(\sigma;\beta)
=
\frac{\sqrt{1-h^2+m_\beta^2}-m_\beta}{2}.
\label{eq:extra-rank-one-CF}
\end{equation}
Product noise is the case $x=\vec a$, $y=\vec b$, which recovers the
product law \eqref{eq:EACF-product}. Eq.~\eqref{eq:extra-rank-one-CF}
extends it to separable states with rank-one $T_\sigma$, including
non-product states. The capacity $\EAC$ is not fixed by $T_\sigma$ alone,
so it has no matching reduction.

\begin{proof}
For separable $\sigma$, the number $h$ is at most one. Indeed, for unit
vectors $u,v\in\mathbb R^3$,
$$
u^TT_\sigma v=\Tr\bigl(\sigma\,(u\cdot\boldsymbol\sigma)\otimes
(v\cdot\boldsymbol\sigma)\bigr),
$$
and the observable on the right has operator norm one. Hence
$|u^TT_\sigma v|\le1$. If $T_\sigma=x\,y^T$, the largest possible value of
$|u^TT_\sigma v|$ is $\|x\|\,\|y\|=h$, so $h\le1$.

The calculation uses the same singular-value step as
Appendix~\ref{subsec:supp-product-tele}. In that step only the numerator
correlation tensor is used, so $\vec a\,\vec b^{\,T}$ can be replaced by
$x\,y^T$. Along the mixing line
$\rho_\eta=(\sigma+\eta\beta)/(1+\eta)$ the numerator correlation tensor
is $M_\eta=\eta D_\beta+x\,y^T$. The fully entangled fraction is evaluated
from the singular values of $M_\eta$. Since $D_\beta$ is orthogonal,
left multiplication by it preserves singular values. Hence $M_\eta$ shares the
singular values of $\eta I+(D_\beta x)y^T$: one equals $\eta$, and on the
teleportation branch $\eta+m_\beta>0$ (where
$\det M_\eta=-\eta^2(\eta+m_\beta)<0$) the other two sum to
$\sqrt{h^2+4\eta(\eta+m_\beta)}$. The
fully entangled fraction is then the $\det<0$ case of
\eqref{eq:supp-horodecki-f-singvals},
$$
f(\rho_\eta)=\frac14\left(1+
\frac{\eta+\sqrt{h^2+4\eta(\eta+m_\beta)}}{1+\eta}\right),
$$
and the onset $f(\rho_\eta)=\tfrac12$ reduces to
$\sqrt{h^2+4\eta(\eta+m_\beta)}=1$, i.e.\
$\eta^2+m_\beta\eta-\tfrac14(1-h^2)=0$, whose nonnegative root is
Eq.~\eqref{eq:extra-rank-one-CF}.
If $\eta+m_\beta\le0$, the same estimate as in the product calculation gives
$f(\rho_\eta)\le \frac14(1+(\eta+h)/(1+\eta))\le\frac12$, since $h\le1$ on
separable states. No crossing is lost by using the
$\eta+m_\beta>0$ branch.
\end{proof}

\end{document}